\documentclass[journal,twosided,a4paper]{IEEEtran}  

\usepackage{amsmath,amssymb,amsfonts,mathtools,enumerate,graphicx,array,tikz,pgfplots,amsthm,mathdots,lipsum,cuted}
\usepackage{enumitem}
\usepackage{varwidth}
\usepackage[linesnumbered,ruled,vlined]{algorithm2e}
\usepackage{tikz}
\usepackage{adjustbox}
\usepackage{stfloats,lipsum} 
\usepackage{multirow,multicol}
\usepackage{hhline,cancel}
\usetikzlibrary{shapes.geometric, arrows, fit,automata, positioning,shapes}
\tikzstyle{startstop} = [rectangle, rounded corners, 
minimum width=3cm, 
minimum height=1cm,
text centered, 
draw=black, 
fill=red!30]

\tikzstyle{io} = [trapezium, 
trapezium stretches=true, 
trapezium left angle=70, 
trapezium right angle=110, 
minimum width=3cm, 
minimum height=1cm, text centered, 
draw=black, fill=blue!30]

\tikzstyle{process} = [rectangle, 
minimum width=3.5cm, 
minimum height=1.15cm, 
text centered, 
text width=3cm, 
draw=black, 
fill=orange!30]

\tikzstyle{decision} = [diamond, 
minimum width=3cm, 
minimum height=1cm, 
text centered, 
draw=black, 
fill=green!30]
\tikzstyle{arrow} = [thick,->,>=stealth]
\tikzstyle{myfit} = [draw,dashed,blue, inner xsep=10pt, inner ysep=15pt, rounded corners=5pt]
\tikzstyle{mytitle}=[draw,densely dashed,blue, fill=red!50, inner sep=5pt, right, yshift=-3.5cm, xshift = -1.825cm]
\newtheorem{defn}{Definition}
\newtheorem{assumption}{Assumption}
\newtheorem{theorem}{Theorem}
\newtheorem{proposition}{Proposition}
\newtheorem{lemma}{Lemma}
\newtheorem{claim}{Claim}
\newtheorem{cor}{Corollary}
\newtheorem{note}{NOTE}

\title{\LARGE \bf
	Fast Randomized Subspace System Identification for Large I/O Data 
}

\author{Vatsal Kedia and Debraj Chakraborty
	\thanks{The authors are with the Department of Electrical Engineering, Indian Institute of Technology	Bombay, Mumbai, Maharashtra, India. {\tt\small Email: \{vatsalkedia, dc\}@ee.iitb.ac.in}}%
}

\begin{document}
	\maketitle
	\begin{abstract}
		In this article, a novel fast randomized subspace system identification method for estimating combined deterministic-stochastic LTI state-space models, is proposed. The algorithm is especially well-suited to identify high-order and multi-scale systems with both fast and slow dynamics, which typically require a large number of input-output data samples for accurate identification using traditional subspace methods. Instead of working with such large matrices, the dataset is compressed using randomized methods, which preserve the range-spaces of these matrices almost surely. A novel identification algorithm using this compressed dataset, is proposed. This method enables the handling of extremely large datasets, which often make conventional algorithms like N4SID, MOESP, etc. run out of computer memory. Moreover the proposed method outperforms these algorithms in terms of memory-cost, data-movement, flop-count and computation time for cases where these algorithms still work in-spite of large data sizes. The effectiveness of the proposed algorithm is established by theoretical analysis and various real and simulated case studies.
	\end{abstract}

	
	\section{INTRODUCTION}
	
	Due to the easy availability of sensor readings and the simultaneous development of highly precise identification algorithms (e.g. see \cite{Ljung_book} and the references therein), data-driven system identification has acquired widespread adoption. Undoubtedly, permanent storage on local hard-drives or on the cloud has become cheap and accessible \cite{Storage}. On the other hand, advances in industrial sensor technology have made long sustained recordings of industrial processes feasible. This has led to wide availability of large amounts of system level input-output data. This data can potentially be used for developing accurate models of the underlying dynamical systems. Conventional system identification algorithms running on personal computers, require to access the stored data by copying it to temporary storage such as random access memory (RAM) and then to processor cache memory (PCM). However, processor caches remain relatively expensive and of limited capacity (see Fig.\ref{fig:memory}). This necessitates frequent transfer of portions of the data between RAM and PCM, thereby degrading algorithm performance. As a result, reducing data transfers is often key to accelerating numerical algorithms in real world \cite{tropp_2020} (eg. see LAPACK \cite{lapack}, BLAS \cite{blas} for modern numerical linear algebra algorithms implementing such optimized transfers).
	
	Input-output data is collected from real-time processes in sampled form. Tuning the sampling frequency and the time period over which the data is collected is a simple way to regulate data size for system identification.  It is known that sampling time plays a very crucial role in identifying the underlying model \cite{Ljung_book}. Conventionally, the sampling frequency is chosen to be around ten times the “guessed" bandwidth of the system (\cite{Ljung_book}, pg. 452). In other words, the sampling frequency is determined by the fastest eigenvalue of the system. On the other hand, the total duration of the collected data is guided by the slowest eigenvalue \cite{Schoukens_explength_2015}.  Hence for unknown systems which might have both very slow, as well as very fast modes, the total number of samples required to identify all the modes becomes very large. For example, in PHWR nuclear reactors the fastest time-constants are in the order of 0.05 seconds, while the slowest oscillations due to Xenon occur over 20 hours \cite{nuclear}. A quick calculation shows that, sampling at 20 times per second for three days (roughly four times the slowest time constant) leads to a single signal producing a $20 \times 60 \times 60 \times 24 \times 3 = 5184000$ sized vector. Other examples exhibiting fast and slow dynamics include, blast furnances \cite{blastfurnace}, reactive distillation columns \cite{distillation_column}, batteries \cite{battries} etc. Sub-sampling and/or reducing the recording duration risks mis-identification of the modes. This results in system identification tasks with necessarily very large input-output data sizes.
	\begin{figure}[h]
		\centering
		\includegraphics[width=0.45\textwidth]{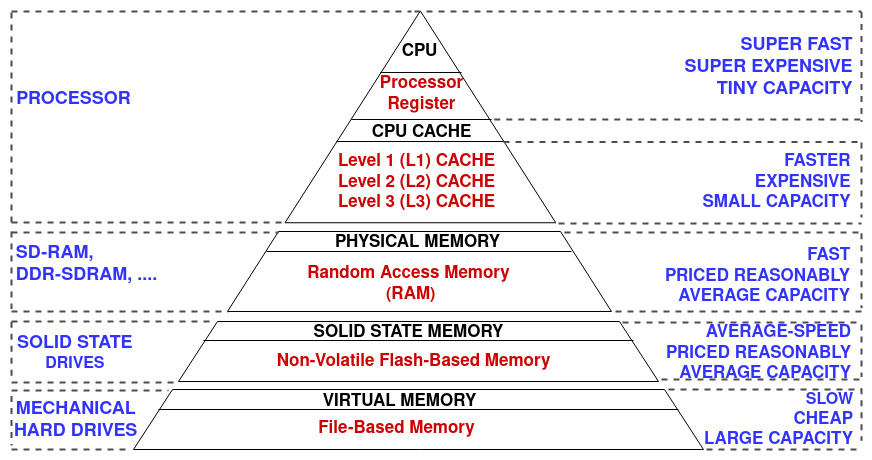}
		\caption{Memory heirarchy \cite{Memoryhierarchy}}
		\label{fig:memory}
	\end{figure}
	
	While various approaches to modeling and identification are used in practice (e.g. see \cite{ljung2010perspectives}
	and the references therein), for large-scale systems, state-space models are preferred. An early example of identification of state space models of linear dynamical systems is the Ho-Kalman algorithm \cite{HoKalman} based on the impulse response data. Subsequently, subspace identification methods based on input-output data were developed (\cite{Ljung_book}, \cite{SubspaceId_book}).   Variations of subspace based methods include Canonical Variate Analyis (CVA) \cite{CVA_Larimore}, Multivariable Output-Error State Space (MOESP) \cite{MOESP_Verhaegan} method and Numerical algorithms for State Space System Identification (N4SID) \cite{N4SID_Overschee_detsto}.
	All subspace identification algorithms are based on finding system related subspaces from I/O data matrices through a series of projections and decompositions. Hence these algorithms require QR and/or SVD decomposition to be performed on various matrices composed of recordings of system inputs and outputs. When dealing with large sample sizes as mentioned above, this leads to prohibitively large time/space complexity as well as increased overhead due to frequent data transfers between PCM and RAM. In systems with moderate RAM capacity, identification of these large and multi-scale systems might encounter ``out-of-memory" issue due to insufficient slow memory (RAM) space.  
	
	Several attempts to increase the efficiency of subspace identification algorithms exist in the literature. Fast subspace identification has been investigated in  \cite{Cho_Kailath_1994}, \cite{Cho_Kailath_1995}, \cite{Sima_2004}, \cite{Dohler_2012} where the focus has been only on the usage of faster QR decomposition methods. The methods proposed have been shown to produce inaccurate estimates for some commonly encountered types of data matrices (e.g. see example 2 in \cite{Sima_2001}). Parameter estimation technique which is independent of data size was attempted in \cite{Katayama_LQfullmodel}. However in all these articles, algorithmic performance was typically characterized by flop counts, while very less or no attention was given on memory usage and data movement, which play crucial roles in determining the computation time for large data matrices (\cite{Demmel_seqQR}, \cite{Demmel_report}). 
	
	Randomized matrix approximation algorithms have gained significant popularity due to their wide applications in large scale data analysis and scientific computing (e.g. see \cite{tropp_2020}, \cite{mahoney_2011}, \cite{halko_2011} and references therein). Subspace identification based on randomized approximation of subspaces has been
	attempted in  \cite{kramer2018_CURHankel}, \cite{minster2021_ERARSVD} and \cite{anderson2022_ESRSVD}.
	All of these papers have assumed the availability of impulse-response data (similar to the Ho-Kalman
	algorithm \cite{HoKalman}). In \cite{kramer2018_CURHankel}, the subspaces were obtained by performing
	SVD on the CUR approximation of the data Hankel matrix. In \cite{minster2021_ERARSVD} and \cite{anderson2022_ESRSVD}, the factorization of the Hankel
	matrix was done by using subspace iteration based randomized SVD.
	In all these articles, significant speedup was obtained as compared to conventional
	Ho-Kalman algorithm while having comparable accuracy in estimates. However, these articles have not suggested any efficient method for estimating the Markov parameters. Consequently the proposed methods are of limited practical applicability.
	
Hence, in this paper, we propose a fast randomized subspace identification algorithm which can handle: (i) the out-of-memory issue and, (ii) extremely large input-output datasets with limited cache capacity at a much higher speed as compared to conventional methods. In the proposed algorithm, we use a randomized range estimation method to approximate the range-space of various data matrices. The primary idea behind the algorithm is to perform iterative/sequential block multiplication between the given (fat/wide) data matrix and a suitably chosen (tall) random matrix of appropriate size. This operation compresses the (large) data matrix to a much smaller (almost square) matrix, while preserving the range space of the original data matrix almost surely (a.s.). Subsequently the typical steps in the subspace identification algorithm such as QR factorization and SVD are performed on the smaller compressed data matrix, instead of the original data matrix. This leads to: 
	(i) improved RAM runtime usage (space complexity), (ii) decreased flop count, and (iii) reduced data transfers between RAM and PCM.   In addition, the estimation of $\{A, B, C, D, K\}$ matrices using conventional methods like N4SID and MOESP depends on the original I/O data size, making it computationally expensive for large data sizes. Hence, a new method to identify $\{B, D, K\}$ matrices is introduced which is independent of the data size and hence takes lesser flops and computation-time as compared to the known methods. Our algorithm successfully identifies the system even when conventional subspace methods such as N4SID and MOESP fail due to out-of-memory issue and outperforms these algorithms in situations where these algorithms can still function in-spite of  large data sizes. Our main contributions are,
	\begin{enumerate}
		\item Combined deterministic-stochastic system identification for large and multi-scale LTI systems with large I/O data sizes. 
		\item Streaming data compression to handle out-of-memory issue.
		\item Fast QR and/or SVD decomposition due to smaller matrix dimension leading to reduced flops and data movement between RAM and PCM. 
		\item Lesser runtime RAM usage due to iterative block update in all computation.
		\item Fast $\{A, B, C, D, K\}$ estimate independent of I/O data size.
	\end{enumerate} 
	Preliminary versions of this work have been published in  \cite{Vatsal_Fastsub} and  \cite{Vatsal_RandSID}. This paper includes the following major extensions:
	\begin{enumerate}[label=(\roman*)]
		\item Identification and detailed analysis of stochastic part including new Theorems and proofs (section \ref{subsection:section_stochasticident}).
		\item Theoretical analysis for the proposed algorithm proving that the  proposed algorithm is fast as compared to conventional algorithms (section \ref{section:algo_performance}) for large data sizes. 
		\item The proposed algorithm is applied to identify model for a real world pressurized heavy water nuclear reactor (PHWR): Zone power variation (section \ref{section_case_studies}). 
	\end{enumerate}
	
	
	\section{Preliminaries and Problem Formulation} 
	We assume that the input $u(t) \in \mathbb{R}^{m}$ and the output $y(t)  \in \mathbb{R}^{p}$ of the following
	$n^{th}$ order discrete time LTI system shown in \eqref{ssmodel},
	are recorded up to $N_{t}$ samples: i.e. $\{u(i),y(i)\}$ $\forall i\in\{0,1,\hdots,N_{t}-1\}$ are recorded.
	\begin{equation}
		\begin{aligned} & x(t+1)=Ax(t)+Bu(t)+Ke(t)\\
			& y(t)=Cx(t)+Du(t)+e(t)
		\end{aligned}
		\label{ssmodel}
	\end{equation}
	Here $e(t)\in\mathbb{R}^{p}$ is known as the innovations process vector
	and is assumed to be a white noise sequence with zero mean and finite
	covariance i.e. $\mathbb{E}\{e(t_{1})e^{T}(t_{2})\}=\eta\delta_{t_{1}t_{2}}$
	for all time instants $t_{1}$ and $t_{2}$, where, $\eta \in \mathbb{R}^{p \times p} > 0$ and
	$\delta$ is the Kronecker delta function. The system parameters $\{A,B,C,D\}$
	are of appropriate dimensions: $A\in\mathbb{R}^{n\times n}$, $B\in\mathbb{R}^{n\times m}$,
	$C\in\mathbb{R}^{p\times n}$, $D\in\mathbb{R}^{p\times m}$ while
	the Kalman gain is denoted by $K\in\mathbb{R}^{n\times p}$.
	
	The objective of any subspace identification algorithm is to estimate
	the model order $n$ and system parameters $\{A,B,C,D,K\}$ up to similarity
	transforms. In the next subsection we briefly review some conventional
	subspace identification algorithms based on \cite{LjungSID}.	
	\subsection{Conventional subspace algorithms}\label{section_conv}
	
	First, a prediction horizon $k$ is chosen such that $k>n$ and the block
	size $N:=N_{t}-2k+2$ is defined. Let us denote a block Hankel matrix based on input sequence $\{u(i)\}$ as follows: 
	\begin{equation*}
		U_{i|i+k-1} :=\begin{bsmallmatrix}u(i) & u(i+1) & \hdots & u(i+N-1)\\
			u(i+1) & u(i+2) & \hdots & u(i+N)\\
			\vdots & \vdots & \ddots & \vdots\\
			u(i+k-1) & u(i+k) & \hdots & u(i+k+N-2)
		\end{bsmallmatrix}
	\end{equation*}
	Now, the past input block Hankel matrix is defined as $U_{p} :=U_{0|k-1} \in \mathbb{R}^{km \times N}$ by substituting $i=0$ above while the corresponding future input matrix is defined as $U_{f} := U_{k|2k-1} \in \mathbb{R}^{km \times N}$ (substituting $i=k$). Similarly we define the output block Hankel matrices $Y_{p},Y_{f}\in\mathbb{R}^{kp\times N}$ using
	the the past/future output data $\{y(i)\}$. Although no recordings of noise are
	assumed to be available, for the sake of notational convenience,
	similar matrices are also defined for the corresponding past and future
	innovations processes: $E_{p},E_{f}\in\mathbb{R}^{kp\times N}$. The
	past input and output data is combined into $W_{p}:=\begin{bmatrix}U_{p}^{T} & Y_{p}^{T}\end{bmatrix}^{T}\in\mathbb{R}^{k(m+p)\times N}$. 
	We further denote, $\Theta_{k} \in \mathbb{R}^{kp\times n}$ as the extended observability matrix,
	$\Psi_{k} \in \mathbb{R}^{kp\times km}$ as impulse response Toeplitz matrix, and $\Phi_{k} \in \mathbb{R}^{kp\times kp}$ as
	noise impulse response Toeplitz matrix as shown below.
	\begin{align*}\label{matrix_structure}
		\Theta_{k} & =\begin{bmatrix}C\\
			CA\\
			CA^{2}\\
			\vdots\\
			CA^{k-1}
		\end{bmatrix} \hspace{0.5cm} \Psi_{k}=\begin{bmatrix}D & 0 & \hdots & 0\\
			CB & D & \ddots & \vdots\\
			\vdots & \ddots & \ddots & 0\\
			CA^{k-2}B & \hdots & CB & D
		\end{bmatrix}\\
		& \Phi_{k}=\begin{bmatrix}I & 0 & \hdots & 0\\
			CK & I & \ddots & \vdots\\
			\vdots & \ddots & \ddots & 0\\
			CA^{k-2}K & \hdots & CK & I
		\end{bmatrix}
	\end{align*}
	\begin{assumption}\label{assumption:system} \cite{LjungSID} The following are assumed: 
		\begin{enumerate}[label=(\roman*)]
			\item The input $u(t)$ is persistently exciting of order $2k$. 
			\item The input $u(t)$ is uncorrelated with innovations $e(t)$. 
			\item No feedback from the output $y(t)$ to the input $u(t)$ exists. 
			\item Eigenvalues of $(A-KC)$ are stable. 
			\item The pair $\{A,C\}$ is observable and the pair $\{A,[B\hspace{0.1cm}K]\}$ is controllable. 
			
		\end{enumerate}
	\end{assumption}
	\begin{defn}
		\label{def:stochasticorthogonalproj}(\cite{SubspaceId_book} pp. 28)
		Let $A\in\mathbb{R}^{p\times j}$, $B\in\mathbb{R}^{q\times j}$ and
		$\mathbb{E}_{j}[.]$ denote the average over one-infinitely long experiment $(j \to \infty)$. In stochastic framework,
		the orthogonal projection of row-space of $A$ onto row-space of $B$
		is defined as 
		\begin{equation}\label{expectation_sto}\begin{aligned}
			A/B&=\mathbb{E}_{j}[AB^{T}]\mathbb{E}_{j}[BB^{T}]^{\dagger}B \approx \bigg[\frac{1}{j}AB^{T}\bigg]\bigg[\frac{1}{j}BB^{T}\bigg]^{\dagger}B\\
			&=[AB^T][BB^T]^{\dagger}B
		\end{aligned}
		\end{equation}
	\end{defn}
	Let $X_{f} \in \mathbb{R}^{n \times N}$ denote the future state sequence defined as $X_{f} := \begin{bmatrix} x(k) & x(k+1)& \hdots & x(k+N-1)\end{bmatrix}$. Using \eqref{ssmodel} recursively and with the data matrices as defined
	above, we get, 
	\begin{equation}
		\begin{aligned} & Y_{f}=\Theta_{k}X_{f}+\Psi_{k}U_{f}+\Phi_{k}E_{f}.\end{aligned}
		\label{Yfuture}
	\end{equation}
	
	Let $\bar{A} := A - KC$, $\bar{B} := B - KD$, $\Upsilon_{k} := \begin{bmatrix} \bar{A}^{k-1}\bar{B} &  \bar{A}^{k-2}\bar{B} & \hdots & \bar{B}\end{bmatrix} \in \mathbb{R}^{n \times km}$ be the modified reversed extended controllability matrix and $\Upsilon_{k}^{e} := \begin{bmatrix} \bar{A}^{k-1}K & \bar{A}^{k-2}K & \hdots & K\end{bmatrix} \in \mathbb{R}^{n \times kp}$
	be the modified reversed extended stochastic controllability matrix. Then, under the assumptions listed above and for large prediction horizons
	$k$, it can be shown \cite{LjungSID} that $X_{f}=L_{p}W_{p}$
	for $L_{p}:=\begin{bmatrix}{\Upsilon}_{k} & {\Upsilon}_{k}^{e}\end{bmatrix}\in\mathbb{R}^{n\times k(m+p)}$. Thereby \eqref{Yfuture} reduces to 
	\begin{equation}
		Y_{f}=\Theta_{k}L_{p}W_{p}+\Psi_{k}U_{f}+\Phi_{k}E_{f}.\label{Yfuturesto}
	\end{equation}
	Subspace algorithms uses orthogonal and/or oblique projections to
	extract subspaces that contains system related information like the extended observability matrix ($\Theta_{k}$) and/or a Kalman state sequence ($X_{f}$).
	One of the methods to accomplish this is to orthogonally project $Y_{f}$
	onto the joint span of $W_{p}$ and $U_{f}$ as follows (using Definition \ref{def:stochasticorthogonalproj}): 
	\begin{equation}\small
		\begin{aligned}Y_{f}/\begin{bsmallmatrix}W_{p}\\
				U_{f}
			\end{bsmallmatrix} & =\Theta_{k}L_{p}W_{p}/\begin{bsmallmatrix}W_{p}\\
				U_{f}
			\end{bsmallmatrix}+\Psi_{k}U_{f}/\begin{bsmallmatrix}W_{p}\\
				U_{f}
			\end{bsmallmatrix}+\Phi_{k}\underbrace{E_{f}/\begin{bsmallmatrix}W_{p}\\
				U_{f}
			\end{bsmallmatrix}}_{= 0 }\\
			& =\Theta_{k}{L}_{p}W_{p}+\Psi_{k}U_{f}
		\end{aligned}
		\label{Yfuturefinal}
	\end{equation}
	The third term in the above equation is zero because $E_{f}$ is
	not correlated with $W_{p}$ and $U_{f}$ in open-loop \cite{LjungSID}.
	
	Now using oblique projections, $Y_{f}$ orthogonally projected onto the joint span of $W_{p}$ and $U_{f}$  can also be written as, 
	\begin{equation}\label{Yf_ortho}
		\begin{aligned}Y_{f}/\begin{bmatrix}W_{p}\\
				U_{f}
			\end{bmatrix} & =Y_{f}/_{U_{f}}W_{p}+Y_{f}/_{W_{p}}U_{f}\\
			& =\underbrace{\bar{L}_{p}W_{p}}_{=:\zeta}+L_{U_{f}}U_{f}
		\end{aligned}
	\end{equation}
	On comparing equation \eqref{Yfuturefinal} and \eqref{Yf_ortho}
	we get $\bar{L}_{p} = \Theta_{k}L_{p}$ and $\zeta :=\bar{L}_{p}W_{p} \in \mathbb{R}^{kp \times N}$. An efficient way to calculate the oblique projection is by using the $QR$ decomposition. 
	\subsubsection{QR step}
	
	Perform LQ decomposition on $H:=\begin{bmatrix}U_{f}^{T} & W_{p}^{T} & Y_{f}^{T}\end{bmatrix}^{T}\in\mathbb{R}^{2k(m+p)\times N}$, or equivalently perform QR decomposition on $H^{T}$, to obtain the decomposition of $Y_{f}$ as shown in \eqref{Yfuturesto}. 
	\begin{equation}\label{Hmatrix}
		\begin{aligned}
			H &= \begin{bmatrix}U_f\\W_{p}\\Y_f\end{bmatrix}
			&= \underbrace{\begin{bmatrix} R_{11} & 0 & 0 \\R_{21} & R_{22} & 0 \\ R_{31} & R_{32} & R_{33} \end{bmatrix}}_{L} \begin{bmatrix}
				Q_{1}^T \\ Q_{2}^T \\Q_{3}^T
			\end{bmatrix}
	\end{aligned} \end{equation}
	From \eqref{Hmatrix},
	\begin{equation}\label{Yf_ortho_data}
		Y_{f}/\begin{bmatrix}W_{p} \\U_{f}\end{bmatrix} = R_{32}R_{22}^{\dagger}W_{p} + (R_{31}-R_{32}R^{\dagger}_{22}R_{21})R^{-1}_{11}U_f
	\end{equation}
	On comparing \eqref{Yf_ortho} and \eqref{Yf_ortho_data},
	\begin{equation}\label{zetaPsik}\begin{aligned}
			& \bar{L}_{p} = R_{32}R_{22}^\dagger
		\end{aligned}
	\end{equation}
	
	Now, it can be shown that using \eqref{Yfuture} and \eqref{Yfuturefinal}, the first terms of $Y_f$ can be equated as, 
	\begin{equation}\label{eqzeta}
		 \underbrace{\Theta_{k} X_f}_{theoretical} = \underbrace{\bar{L}_{p}W_{p}}_{data} =: \zeta.
	\end{equation} 
	In practical implementations of the QR factorization for large matrices (eg. MATLAB, LAPACK \cite{lapack}, etc.),  a sequential version of the QR algorithm \cite{Demmel_report}, which partitions the data into (say) $d$ sized blocks and iteratively computes the Q and R factors, is used.
	\subsubsection{SVD step}
	
	Next we calculate the SVD of $\zeta$ as follows: 
	\begin{equation}
		\begin{aligned}\zeta & =\begin{bmatrix}U_{1} & U_{2}\end{bmatrix}\begin{bmatrix}\Sigma_{1} & 0\\
				0 & \Sigma_{2}
			\end{bmatrix}\begin{bmatrix}V_{1}^{T}\\
				V_{2}^{T}
			\end{bmatrix}\\
			& =U_{1}\Sigma_{1}V_{1}^{T}+\underbrace{U_{2}\Sigma_{2}V_{2}^{T}}_{noise}\\
			& \approx U_{1}\Sigma_{1}V_{1}^{T}=\underbrace{U_{1}\Sigma_{1}^{1/2}}_{{\hat{\Theta}}_{k}}\underbrace{\Sigma_{1}^{1/2}V_{1}^{T}}_{\hat{X}_{f}}
		\end{aligned}
		\label{eqzetasvd}
	\end{equation}
	The second term is ignored assuming that the noise component is negligible
	as compared to the system contribution. The order of the system is
	determined from $rank(\zeta)$ = $rank(\hat{\Theta}_{k})$ = $n$ where $\hat{\Theta}_{k} = U_{1}\Sigma_{1}^{1/2} \in \mathbb{R}^{kp \times n}$ and $\hat{X}_{f} = \Sigma_{1}^{1/2}V_{1}^{T} \in \mathbb{R}^{n \times N}$. For large matrices, modern SVD implementations use block QR algorithms (eg. see LAPACK's dgesvd algorithm \cite{lapack}). Let $\mathcal{R}(.)$ denote the range-space of a matrix. Then:
	\begin{proposition}
		\label{Lemma_zetathetak} \cite{SubspaceId_book} $\mathcal{R}(\zeta)=\mathcal{R}(\Theta_{k})$. 
	\end{proposition}
	
	\subsection{Estimating system matrices} \label{system_matrices}
	A brief overview of the N4SID and MOESP class of algorithms \cite{SubspaceId_book} is presented below. 
	\subsubsection{N4SID}(\cite{SubspaceId_book}, \cite{N4SID_Overschee_detsto},  \cite{Katayama_book})
	In this method, the system parameters are found by using the estimate of the state sequence matrix $X_{f}$ introduced in \eqref{eqzetasvd} (denoted by $\hat{X}_f$) and solving the stacked state space equations in the least squares sense. 
	\subsubsection{MOESP}
	This method uses the estimate of the extended observability matrix $\hat{\Theta}_{k}$ (see \eqref{eqzetasvd}) to estimate the $\{A,C\}$ pair: $A$ is estimated using the shift invariance property of $\Theta_{k}$ as,
	\begin{equation}\label{eqAmatrix}\begin{aligned}
			\hat{A} = arg \underset{A \in \mathbb{R}^{n \times n}}{min} || \hat{\Theta}^{\downarrow}_{k}A - \hat{\Theta}^{\uparrow}_{k}||_{F}= (\hat{\Theta}_{k}^{\downarrow})^{\dagger}\hat{\Theta}_{k}^{\uparrow}
		\end{aligned}	
	\end{equation} 
	and $\hat{C} = \hat{\Theta}_{k}(1:p,:)$, where, $\hat{\Theta}_{k}^{\downarrow} := \hat{\Theta}_{k}(1:p(k-1),:)$ and $\hat{\Theta}_{k}^{\uparrow} := \hat{\Theta}_{k}(p+1:kp,:)$. The estimation of $\{B, D,K\}$ for MOESP class of algorithms makes use of full data of size $N$. For example, in \cite{Ljung_book}, an approach to estimate $\{B, D\}$ is given by
	\begin{equation}
		arg \hspace{0.1cm}\underset{B,D}{min} \frac{1}{N} \sum_{t=1}^{N} || y(t) - \hat{C}(qI - \hat{A})^{-1}Bu(t) - Du(t)||
	\end{equation}
	Several other approaches to estimate $\{B, D, K\}$ can be found in \cite{Ljung_book}, \cite{MOESP_Verhaegan}, \cite{Ljung_MATLAB}, \cite{MOESP_bd_Verhaegan} and \cite{MOESP_Verhaegen_stochastic}.
	In summary, the major steps of conventional subspace identification are presented in Algorithm \ref{algo:conv_SID}.
	\begin{algorithm}\label{algo:conv_SID}
		\textbf{Input:} Load input-output data into RAM.\\
		Formulate data matrices from input-output data in RAM. \\
		Perform LQ decomposition on $H = \frac{1}{\sqrt{N}}\begin{bmatrix}
			U_{f}^{T} & U_{p}^{T} & Y_{p}^{T} & Y_{f}^{T}
		\end{bmatrix}^{T} \in \mathbb{R}^{2k(m+p) \times N}$ (see \eqref{Hmatrix}). Use sequential QR if $H$ is large.\\
		Perform SVD on $\zeta = \bar{L}_pW_{p} \in \mathbb{R}^{kp \times N}$ to estimate $\Theta_{k}$ and/or $X_{f}$ (see \eqref{eqzetasvd}). Use sequential SVD if $\zeta$ is large.\\
		Estimate $\{A, B, C, D, K\}$ using N4SID or MOESP type algorithms (refer section \ref{system_matrices}).\\
		\textbf{Output:} Estimated $\{A, B, C, D, K\}$.
		\caption{Conventional subspace identification}
	\end{algorithm}
	
	\begin{assumption}\label{assumption:data} Throughout this paper
		we assume: 
		\begin{enumerate}[label=(\roman*)]
			\item Input-output data (i.e. $U \in \mathbb{R}^{m \times N_{t}}$ and $Y \in \mathbb{R}^{p \times N_{t}}$) fits into RAM (slow memory). We ignore the space occupied by the raw I/O data (U and Y) in all computations related to space-complexity in this paper.
			\item There are no sub-levels in cache memory.
		\end{enumerate}
	\end{assumption}
	\begin{note}\label{Case:RAM_space}
		In this article we consider two cases depending on size of 	the matrix $H \in \mathbb{R}^{2k(m+p) \times N}$ with $N >> 2k(m+p)$:
		\begin{enumerate}
			\item  The matrix $H$ fit into RAM but not in cache memory.
			\item The matrix $H$ does not fit into RAM (out-of-memory).
		\end{enumerate} 
	\end{note}
	The case 2 above is often encountered in many modern applications, such as high-dimensional systems and systems with widely separated time-scales, as outlined in the introduction. 
	\subsection{Algorithm performance}
	We will evaluate the algorithm performance based on a combination of the following performance metrics. 
	\begin{enumerate}
		\item Memory-cost (denoted by $M$) is defined as the RAM space needed for the algorithm to identify the system parameters (see Assumption \ref{assumption:data}). Let one unit of space be required to store one word \cite{Demmel_seqQR}. Then, e.g., the matrix $A \in \mathbb{R}^{m \times n}$ requires $mn$ units of space.
		\item Flop count (denoted by $F$) is defined as the number of floating point operations i.e. the number of additions and multiplications required to estimate the model.
		\item Data movement \cite{Demmel_report} (denoted by DM) between RAM (slow memory) and PCM (fast memory) during a run of the algorithm: This is characterized by two variables namely, $(a)$ $\#words$, denoting the total number of read and written words between RAM and PCM, and $(b)$ $\#messages$, denoting the total data packets moved between RAM and PCM.
	\end{enumerate}
	The actual computation-time can be expressed as (\cite{Demmel_seqQR}, \cite{Demmel_report}),
	\begin{equation}\label{Talgo}
		T_{algo} = F \times \gamma + \#messages \times \alpha + \#words \times \beta
	\end{equation} 
	where, $\gamma$ denotes time per flop, $\alpha$ denotes latency and $\beta$ is inverse of the memory bandwidth. We assume read and write bandwidth between slow and fast memory are same \cite{Demmel_seqQR}. The last two terms of the above equation constitute the communication time between slow (RAM) and fast memory (PCM). However, $\alpha, \beta$ and $\gamma$ are machine dependent and not known in most real-life situations.  Hence, we use the following proxy representing a combination of computation-time costs as well as memory-cost.
	
	\begin{defn}
	    The algorithm cost ($C$) is defined as:
	\begin{equation}\label{Algo_cost}
		C = M + F + \underbrace{(\#words + \#messages)}_{=: DM}
	\end{equation}
	\end{defn}
	
	\subsection{Main issues and Problem Formulation}\label{section_main_issues}
	For the conventional methods (see section \ref{section_conv} and \ref{system_matrices}), the cost defined in \eqref{Algo_cost} can increase substantially for high-dimensional system and large $N$:
	\begin{enumerate} \item Memory-cost ($M_{conv}$): The space required for QR decomposition (see \eqref{Hmatrix}) is $2k(m+p)N$ units, to store the matrix $W_{p}$ is $k(m+p)N$ units and to store $\zeta$ is $kpN$ units (see Algorithm \ref{algo:conv_SID}). Adding all of them, we get, 
		\begin{equation}\label{Mconv}
			M_{conv} \approx k(3m+4p)N
		\end{equation} 
		We have not considered the memory space for the matrices that does not contains $N$ as one of its dimensions.
		\item Flop-count ($F_{conv}$): The number of flops for conventional algorithms comprises of the flops associated with the three major computation steps as shown in table \ref{tab:flopcount_conv}. 
		\begin{table}[h]
			\setlength\extrarowheight{2pt}
			\begin{center}
				\caption{Flop-count for Conventional method}
				\label{tab:flopcount_conv}
				\scalebox{1}{
					\begin{tabular}{|l|l|}
						\hline
						\multicolumn{1}{|c|}{\textbf{Algorithm steps}} & \multicolumn{1}{c|}{\textbf{Flop-count}}                  \\ \hline
						\begin{tabular}[c]{@{}l@{}}QR on $H^{T}$   \end{tabular}          &   \begin{tabular}[c]{@{}l@{}}$8k^2(m+p)^2N - \frac{16}{3}k^3(m+p)^3$ \end{tabular}          \\ \hline
						\begin{tabular}[c]{@{}l@{}}Matrix mul. $(L_{p}W_{p})$ \\+ SVD on $\zeta$ \end{tabular}                      &          \begin{tabular}[c]{@{}l@{}}$\underbrace{2k^2p(m+p)N}_{\text{mat mul.}} + \underbrace{2k^2p^2N + 2k^3p^3}_{\text{SVD}}$ \end{tabular}                                   \\ \hline
						\begin{tabular}[c]{@{}l@{}}Estimating \\$\{A, B, C, D, K\}$ \end{tabular}                            &          $\approx (4n^2+(6m+2p)n+2m^2+2pm)N$                             \\ \hline
				\end{tabular}}
			\end{center}
		\end{table}  
		
		$F_{ABCDK}$: The N4SID algorithm uses least squares with $N$ equations and estimation of $\{B, D, K\}$ by MOESP approach uses full data i.e. $N$ equations, leading to approximately $\mathcal{O}(n^2N)$ computations \cite{Ljung_book}, \cite{MOESP_Verhaegen_stochastic}. Therefore, using Table \ref{tab:flopcount_conv} above we can write: \begin{equation}\label{FconvABCDK}\begin{aligned}
				F_{conv}  &= F_{QR} +F_{SVD} + F_{ABCDK}\\
		&\approx \big(8k^2(m+p)^2+2k^2p(m+p)+2k^2p^2\\&+4n^2+(6m+2p)n+2m^2+2pm \big)N\\&+2k^3p^3-(16/3)k^3(m+p)^3
			\end{aligned}
		\end{equation}
		\item Data moved: The $\#words$ and $\#messages$ moved will be incurred mainly due to the (sequential) QR and/or SVD steps. In all conventional methods, data movement for QR decomposition is lower bounded by $\frac{k^2N}{\sqrt{W}}$ (see \eqref{Hmatrix} and Appendix $A$ in \cite{Demmel_seqQR}), where $W$ denote the size of fast memory (in terms of floating point words).
	\end{enumerate} 
	
	From the above analysis, it is evident that the algorithm cost (C) defined above increases at-least linearly with $k$ and $N$ and can be prohibitively large for applications with large multi-scale systems (see Introduction). To address this issue we formulate the following problem:\\
	\textbf{Problem 1.} Design a streaming randomized system identification algorithm which reduces cost (C) defined in \eqref{Algo_cost}. This algorithm should simultaneously minimize the error between predicted and actual output. 
 
	

	\section{Preliminary Results: Matrix Range Approximation}\label{section:RNLAtools} 
	The key idea behind the proposed algorithm
	is approximating the range-spaces of various data matrices. This approximations
	are achieved via right matrix multiplication operation of the original
	matrices (fat/wide) with random gaussian iid matrices (tall). In this section, we guarantee
	the preservation of the range-space of a matrix a.s. under
	such an operation. 
	\begin{defn}
		\cite{bryc_rotation}, \cite{normal_rotation} A probability density function $g(.)$ over
		$\mathbb{R}^{n}$ is said to be rotationally invariant, if $\forall X\in\mathbb{R}^{n}$
		and for all rotation matrices $R\in\mathbb{R}^{n\times n}$ such that $RR^{T} = R^{T}R = I_{n}$,
		$g(RX)=g(X)$. 
	\end{defn}
	\begin{lemma}
		\label{Lemma_rotationinvariant}\cite{normal_rotation} Let $X\in\mathbb{R}^{n}$
		such that each $x_{i}\sim\mathcal{N}(0,1)$ is chosen independently
		$\forall i\in\{1,2,\hdots,n\}$. Then the joint density function of
		$X$ is rotationally invariant. 
	\end{lemma}
	Let $Z=\begin{bmatrix}Z_{1} & Z_{2} & \hdots & Z_{m}\end{bmatrix}\in\mathbb{R}^{n\times m}$
	such that $Z_{k}\in\mathbb{R}^{n}$ for $k\in[1,2,...,m]$, and $z_{ij}\sim\mathcal{N}(0,1)$
	are chosen independently $\forall i,j$. If $Z_{v}:=\begin{bmatrix}Z_{1}^{T} & Z_{2}^{T} & \hdots & Z_{m}^{T}\end{bmatrix}^{T}\in\mathbb{R}^{mn}$,
	then trivially $Z_{v}\sim\mathcal{N}(0_{mn},I_{mn})$. Define, $Z_{R}:=R^{T}Z$
	such that $R\in\mathbb{R}^{n\times n}$, $R^{T}R=RR^{T}=I_{n}$ and
	$Z_{Rv}:=\begin{bmatrix}RZ_{1}^{T} & RZ_{2}^{T} & \hdots & RZ_{m}^{T}\end{bmatrix}^{T}\in\mathbb{R}^{mn}$ 
	\begin{lemma}
		\label{Lemma_matrix_normalrotation} $Z_{Rv}\sim\mathcal{N}(0_{mn},I_{mn})$. 
	\end{lemma}
	\begin{proof}
		It is easy to see that $Z_{Rv}=diag(R^{T},R^{T},\hdots,R^{T})Z_{v}=: \tilde{R}^{T}Z_{v}$
		where $\tilde{R}\in\mathbb{R}^{mn\times mn}$. Therefore using Lemma
		\ref{Lemma_rotationinvariant}, $Z_{Rv}\sim\mathcal{N}(0_{mn},I_{mn})$. 
	\end{proof}
	The next Lemma is similar to Theorem 1 in \cite{feng_randomrank}. 
	\begin{lemma}
		\label{Lemma_iid_nn} \label{Cor_iid_mn} Let $A\in\mathbb{R}^{m\times n}$
		be a random matrix whose entries are iid gaussian: $a_{ij}\sim\mathcal{N}(\mu,\,\sigma^{2})$.
		Then $A$ has full rank a.s. i.e. $\mathbb{P}\{rank(A)=min\{m,n\}\}=1$. 
	\end{lemma}
	\begin{cor}
		\label{Lemma_iid_submat} Let $A\in\mathbb{R}^{m\times n}$ be a random
		matrix whose entries are iid gaussian: $a_{ij}\sim\mathcal{N}(\mu,\,\sigma^{2})$.
		Then any square submatrix $A_{k}\in\mathbb{R}^{k\times k}$ of $A$
		has full rank a.s. i.e. $\mathbb{P}\{rank(A_{k})=k\}=1$. 
	\end{cor}
	\begin{proof}
		Since the elements of any submatrix $A_{k}$ are iid gaussian with
		$a_{ij}\sim\mathcal{N}(\mu,\,\sigma^{2})$, Lemma \ref{Lemma_iid_nn}
		ensures that $rank(A_{k})=k$ a.s. 
	\end{proof}
	\begin{lemma}
		\label{Theorem_rank_preservation} Let $A$ be any $\mathbb{R}^{m\times n}$
		matrix and $B\in\mathbb{R}^{n\times p}$ be a random matrix whose
		elements are iid gaussian with zero mean and unity variance i.e. $b_{ij}\sim\mathcal{N}(0,1)$
		with $m<p<n$. Then $rank(AB)=rank(A)$ a.s. i.e. $\mathbb{P}\{rank(AB)=rank(A)\}=1$. 
	\end{lemma}
	\begin{proof}
		Let $rank(A)=k\leq m$. Then the SVD of $A$ can be written as $A=U\Sigma V^{T}$,
		where $\Sigma=\begin{bmatrix}\Sigma_{k} & 0\\
			0 & 0
		\end{bmatrix}$ with $\Sigma_{k}\in\mathbb{R}^{k\times k},V^{T}\in\mathbb{R}^{n\times n}$.
		Then from Lemma \ref{Lemma_matrix_normalrotation}, $\text{\ensuremath{\tilde{B}:=V^{T}B} }$
		is a random matrix whose elements are iid gaussian with $\tilde{b}_{ij}\sim\mathcal{N}(0,1)$,
		Now,
		
		\[
		\begin{aligned}AB & =(U\Sigma V^{T})B\\
			& =U\Sigma\tilde{B}\\
			& =U\begin{bmatrix}\Sigma_{k} & 0\\
				0 & 0
			\end{bmatrix}\tilde{B}\\
			& =U\begin{bmatrix}\Sigma_{k} & 0\\
				0 & 0
			\end{bmatrix}\begin{bmatrix}\tilde{B}_{1} & \tilde{B}_{2}\\
				\tilde{B}_{3} & \tilde{B}_{4}
			\end{bmatrix}\hspace{0.5cm}\text{(where, \ensuremath{\tilde{B}_{1}\in\mathbb{R}^{k\times k}})}\\
			& =U\begin{bmatrix}\Sigma_{k}\tilde{B}_{1} & \Sigma_{k}\tilde{B}_{2}\\
				0 & 0
			\end{bmatrix}
		\end{aligned}
		\]
		From Lemma \ref{Cor_iid_mn}, $rank(B)=p$. It follows that $rank(\tilde{B})=p$
		since $V$ is invertible. From the above calculation, $rank(AB)=rank(\Sigma_{k}\tilde{B}_{1})=rank(\tilde{B}_{1})$.
		Since both $U$ and $\Sigma_{k}$ are invertible, using corollary
		\ref{Lemma_iid_submat}, $rank(AB)=k$ a.s. 
	\end{proof}
	\begin{theorem}
		\label{Theorem_image_presevation} Let $A$ be any $\mathbb{R}^{m\times n}$
		matrix and $B\in\mathbb{R}^{n\times p}$ be a random matrix whose
		elements are iid gaussian with zero mean and unity variance i.e. $b_{ij}\sim\mathcal{N}(0,1)$
		with $m<p<n$. Then $\mathcal{R}(A)=\mathcal{R}(AB)$ a.s. 
	\end{theorem}
	\begin{proof}
		Firstly, $\mathcal{R}(AB)\subseteq\mathcal{R}(A)$ follows trivially.
		Moreover, $rank(AB)=rank(A)$ a.s. from Lemma \ref{Theorem_rank_preservation}.
		Therefore, $\mathcal{R}(A)=\mathcal{R}(AB)$ a.s. 
	\end{proof}

	
	\section{Fast Randomized Subspace System Identification (FR2SID)}\label{fast_subspace}
	In order to address the issues discussed in section \ref{section_main_issues}, we propose an algorithm based on randomized column space preservation (see Theorem \ref{Theorem_image_presevation}) and a novel method to estimate $\{B,D,K\}$ matrices. In the proposed method $\{A, B, C, D, K\}$ are estimated using compressed I/O data. Consequently, while the effort required for the initial compression is still dependent on $N$, all the subsequent computations becomes independent of $N$. In addition, we utilize the traditional power method to enhance noise robustness.
	
	\subsection{Streaming Data Compression: Range-space approximation}\label{subsection:SDC}	
	Recall that the matrix $H:=\begin{bmatrix}U_{f}^{T} & W_{p}^{T} & Y_{f}^{T}\end{bmatrix}^{T}\in\mathbb{R}^{2k(m+p)\times N}$
	and define $N_{c}:=2k(m+p)+l$ where, $l>0$ is commonly known as
	the oversampling parameter \cite{halko_2011}. Define $\mathcal{C}\in\mathbb{R}^{N\times N_{c}}$
	to be a random matrix whose elements are iid gaussian with $(\mathcal{C})_{ij}\sim\mathcal{N}(0,\frac{1}{N_{c}})$. Let $q \in \{0,1\}$,  we further define,
	$\bar{U}_{f}:=(U_{f}U_{f}^T)^{q}U_{f}\mathcal{C}$, $\bar{U}_{p}:=(U_{p}U_{p}^T)^{q}U_{p}\mathcal{C}$, $\bar{Y}_{p}:=(Y_{p}Y_{p}^T)^{q}Y_{p}\mathcal{C}$,
	$\bar{Y}_{f}:=(Y_{f}Y_{f}^T)^{q}Y_{f}\mathcal{C}$, $\bar{E}_{p}:=(E_{p}E_{p}^T)^{q}E_{p}\mathcal{C}$, $\bar{E}_{f}:=(E_{f}E_{f}^T)^{q}E_{f}\mathcal{C}$,
	$\bar{W}_{p}:=(W_{p}W_{p}^T)^{q}W_{p}\mathcal{C}=\begin{bmatrix}\bar{U}_{p}^{T} & \bar{Y}_{p}^{T}\end{bmatrix}^{T}$
	and $\bar{H}:=(HH^T)^{q}H\mathcal{C}=\begin{bmatrix}\bar{U}_{f}^{T} & \bar{W}_{p}^{T} & \bar{Y}_{f}^{T}\end{bmatrix}^{T}\in\mathbb{R}^{2k(m+p)\times N_{c}}$. The next Lemma directly follows from Theorem \ref{Theorem_image_presevation}.
	\begin{lemma}
		\label{Lemma_iosubspacepreserved} $\mathcal{R}(U_{f})=\mathcal{R}(\bar{U}_{f})$,
		$\mathcal{R}(E_{f})=\mathcal{R}(\bar{E}_{f})$, $\mathcal{R}(U_{p})=\mathcal{R}(\bar{U}_{p})$,
		$\mathcal{R}(E_{p})=\mathcal{R}(\bar{E}_{p})$ and $\mathcal{R}(Y_{f})=\mathcal{R}(\bar{Y}_{f})$
		a.s. 
	\end{lemma}
	
	Let us define $N_{d} := N/d \in \mathbb{Z}_{+}$, where $d \in \mathbb{Z}_{+}$, and recall that $W$ is the size of fast memory. The following assumption ensures that the partitioned matrices fit into fast memory.
	\begin{assumption}\label{assumption:W_Nd}
		We choose $d$ such that $W \geq 4k(m+p)N_{d}+lN_{d}+(q+1)4k^2(m+p)^2+2kl(m+p)$.
	\end{assumption}
	\begin{itemize}
		\item Streaming Data Compression (SDC) Algorithm: The (large) data matrix $H$ is divided into $d$ block matrices satisfying assumption \ref{assumption:W_Nd} and then random matrices $\mathcal{C}_{i} \in \mathbb{R}^{N_{d} \times N_{c}}$  are generated sequentially for all $i \in \{1, 2, \hdots, d\}$. The steps for streaming data compression are summarized in Algorithm \ref{algo:fastdatapre}.
			\end{itemize}
	\begin{algorithm}[h]\label{algo:fastdatapre}
		\textbf{Input:} I/O training data i.e. $\{u(i)\}^{N_{t}-1}_{i=0}$ and $\{y(i)\}^{N_{t}-1}_{i=0}$.\\
		Choose hyper-parameters: $q$, $k$, $d$ and $l$ such that $q \in \{0,1\}$, $k >n$, $N_{d} = N/d \in \mathbb{Z}_{+}$, $N_{c} = 2k(m+p)+l$ and Assumption \ref{assumption:W_Nd} is satisfied.\\
		$i \gets 1$\;
		\While{$i \neq {d}$}{
			Formulate data matrices iteratively: $U_{p_{i}} \in \mathbb{R}^{km \times N_{d}}$, $U_{f_{i}}\in \mathbb{R}^{km \times N_{d}}$, $Y_{p_{i}}\in \mathbb{R}^{kp \times N_{d}}$ and $Y_{f_{i}}\in \mathbb{R}^{kp \times N_{d}}$.\\
			Generate random Gaussian iid matrix $\mathcal{C}_{i} \in \mathbb{R}^{N_{d} \times N_{c}}$.\\
			Perform matrix multiplication to compute: $\bar{U}_{p}= \bar{U}_{p}+(U_{p_{i}}U_{p_{i}}^T)^{q}U_{p_{i}}\mathcal{C}_{i}$,
			$\bar{U}_{f}=\bar{U}_{f} + (U_{f_{i}}U_{f_{i}}^T)^{q}U_{f_{i}}\mathcal{C}_{i}$, $\bar{Y}_{p}=\bar{Y}_{p} + (Y_{p_{i}}Y_{p_{i}}^T)^{q}Y_{p_{i}}\mathcal{C}_{i}$ and $\bar{Y}_{f}=\bar{Y}_{f} +  (Y_{f_{i}}Y_{f_{i}}^T)^{q}Y_{f_{i}}\mathcal{C}_{i}$.\\
			$i \gets i+1$\;
		}
		\textbf{Output:} $\bar{U}_{p}$, $\bar{U}_{f}$, $\bar{Y}_{p}$ and $\bar{Y}_{f}$. (Equivalently $\bar{H}$)
		\caption{Streaming Data Compression (SDC)}
	\end{algorithm}

It is important to note that, the matrix multiplication of $\bar{H} = (HH^T)^{q}H\mathcal{C}$ can be easily parallelized.

	\subsection{Projection: QR Step}\label{subsection:QRstep}
	The first step in conventional subspace identification is to perform
	QR on $H^{T}$ (see step 3 Algorithm \ref{algo:conv_SID}) to obtain the oblique projection $\zeta$. Since, we need only the $R$-factor from the QR step, we propose to perform QR on the compressed matrix $\bar{H}^{T} \in \mathbb{R}^{N_{c} \times 2k(m+p)}$. This step reduces the QR computation cost significantly since $N_{c}<<N$. Also, $\bar{H}$ is guaranteed to fit into the fast memory due to Assumption \ref{assumption:W_Nd}. Hence data movement between slow and fast memory is reduced, in turn leading to faster QR implementation as compared to QR on $H$. However for this method to work, we must show theoretically that appropriate projections can still be used to extract the desired subspaces even after data compression.
 
   First note that \eqref{Yfuturesto} can be right multiplied by $\mathcal{C}$  to yield
	\begin{equation}
		\bar{Y}_{f}=\bar{L}_{p}\bar{W}_{p}+\Psi_{k}\bar{U}_{f}+\Phi_{k}\bar{E}_{f}\label{eqYfinalOmega}
	\end{equation}
	Evidently, $\bar{L}_{p}$ remains the same as the uncompressed case, since the multiplication from right by $\mathcal{C}$  does not affect $\bar{L}_{p}$ . Now, the orthogonal projection of $\bar{Y}_{f}$ onto the joint span of $\bar{W}_{p}$ and $\bar{U}_{f}$ is 
	\begin{equation}
		\begin{aligned}\label{Ybarfinal}\bar{Y}_{f}/\begin{bmatrix}\bar{W}_{p}\\
				\bar{U}_{f}
			\end{bmatrix} & =\bar{L}_{p}\bar{W}_{p}/\begin{bmatrix}\bar{W}_{p}\\
				\bar{U}_{f}
			\end{bmatrix}+\Psi_{k}\bar{U}_{f}/\begin{bmatrix}\bar{W}_{p}\\
				\bar{U}_{f}
			\end{bmatrix}\\
			& +\Phi_{k}\bar{E}_{f}/\begin{bmatrix}\bar{W}_{p}\\
					\bar{U}_{f}
			\end{bmatrix}
		\end{aligned}
	\end{equation}
 From \eqref{Yfuturefinal}, we know $E_{f}/\begin{bmatrix}W_{p}\\U_{f}\end{bmatrix}=0$. Now we show that even after data compression using $\mathcal{C}$, the compressed innovation $\bar{E}_{f}$ remains approximately uncorrelated with compressed data $\bar{W}_{p}$ and $\bar{U}_{f}$. 
	
	Let us define columns of ${E}_{f}$ as $e_{i}\in \mathbb{R}^{kp}$ so that ${E}_{f} := \begin{bmatrix}e_{k} & e_{k+1} & \hdots & e_{k+N-1}\end{bmatrix}$ and similarly ${U}_{f} := \begin{bmatrix}u_{k} & u_{k+1} & \hdots & u_{k+N-1}\end{bmatrix}$. Let $\mathcal{C} \in \mathbb{R}^{N \times N_{c}}$ be defined as in section \ref{subsection:SDC} and further let the rows of the random matrix $\mathcal{C}$ be denoted as $\alpha_{i}^{T} \in \mathbb{R}^{1 \times N_{c}}$
	so that $\mathcal{C} :=\begin{bmatrix}\alpha_{k}^{T}\\
		\alpha_{k+1}^{T}\\
		\vdots\\
		\alpha_{k+N-1}^{T}
	\end{bmatrix}\in\mathbb{R}^{N\times N_{c}}$ and define $\beta_{ij}:=(\mathcal{C}\mathcal{C}^T)_{ij} = \alpha_{k+i-1}^{T}\alpha_{k+j-1} \hspace{0.15cm} \forall i,j \in \{1, 2, \hdots, N\}$. It is easy to see that, the diagonal elements of matrix $\mathcal{C}\mathcal{C}^T \in \mathbb{R}^{N \times N}$ follows $\chi^2$-distribution implies $\mathbb{E}[(\mathcal{C}\mathcal{C}^T)_{ii}] = 1$ and $\mathbb{E}[(\mathcal{C}\mathcal{C}^T)_{ij}] = 0$ for $i \neq j$. We assume that sample mean is approximately equal to the expected value. Therefore, \begin{equation}\label{eq_BBT_expectation} 
		\beta_{ij} \approx \begin{cases}
			1 & i = j\\
			0 & i \neq j \hspace{0.2cm} (iid)
	\end{cases}\end{equation}
 Next we present an intermediate result for characterizing the correlation between noise, input and output data. 
	\begin{lemma}
		\label{Lemma_Efbar} $\bar{E}_{f}\bar{U}_{f}^{T} = S_{E_{f}}^{q}\big(\sum_{i=1}^{N}\sum_{j=1}^{N}\beta_{ij}e_{k+i-1}u_{k+j-1}^{T}\big)S_{U_{f}}^{q}$ where $S_{E_{f}} = \sum_{s =1}^{N}e_{k+s-1}e_{k+s-1}^{T}$ and $S_{U_{f}} = \sum_{r=1}^{N}u_{k+r-1}u_{k+r-1}^{T}$.
	\end{lemma}
	\begin{proof}
		From definition of $\bar{E}_{f}$ and $\bar{U}_{f}$ we get:
		\begin{equation}\label{eq_EfUfbar}	\begin{aligned}\bar{E}_{f}\bar{U}_{f}^{T} &= [(E_{f}E_{f}^T)^{q}E_{f}\mathcal{C}] [(U_{f}U_{f}^T)^{q}U_{f}\mathcal{C}]^{T}\\
		&= \underbrace{(E_{f}E_{f}^T)^{q}}_{=:S_{E_f}^{q}}(E_{f}\mathcal{C})(U_{f}\mathcal{C})^{T} \underbrace{(U_{f}U_{f}^T)^{q}}_{=:S_{U_f}^{q}}\\
			\end{aligned}\end{equation}
		Using $E_{f}$ and $\mathcal{C}$:
		\[
		\begin{aligned}E_{f}\mathcal{C} &=\begin{bmatrix}e_{k} & e_{k+1} & \hdots & e_{k+N-1}\end{bmatrix}\begin{bmatrix}\alpha_{k}^{T}\\
				\alpha_{k+1}^{T}\\
				\vdots\\
				\alpha_{k+N-1}^{T}
			\end{bmatrix}\\
			& =e_{k}\alpha_{k}^{T}+e_{k+1}\alpha_{k+1}^{T}+\hdots+e_{k+N-1}\alpha_{k+N-1}^{T}
		\end{aligned}
		\]
		and $\begin{aligned}{U}_{f}\mathcal{C}=u_{k}\alpha_{k}^{T}+u_{k+1}\alpha_{k+1}^{T}+\hdots+u_{k+N-1}\alpha_{k+N-1}^{T}\end{aligned}
		$. Then the middle terms of \eqref{eq_EfUfbar} can be written
		as: 
		\[
		\begin{aligned}({E}_{f}\mathcal{C})({U}_{f}\mathcal{C})^{T} & =(e_{k}\alpha_{k}^{T}+e_{k+1}\alpha_{k+1}^{T}+\hdots+e_{k+N-1}\alpha_{k+N-1}^{T})\\
			& \times (u_{k}\alpha_{k}^{T}+u_{k+1}\alpha_{k+1}^{T}+\hdots+u_{k+N-1}\alpha_{k+N-1}^{T})^{T}\\
			& =\sum_{i=0}^{N-1}\sum_{j=0}^{N-1}(e_{k+i}\alpha_{k+i}^{T})(u_{k+j}\alpha_{k+j}^{T})^{T}\\
			& =\sum_{i=1}^{N}\sum_{j=1}^{N}e_{k+i-1}(\alpha_{k+i-1}^{T}\alpha_{k+j-1})u_{k+j-1}^{T}\\
			& =\sum_{i=1}^{N}\sum_{j=1}^{N}\beta_{ij}e_{k+i-1}u_{k+j-1}^{T}
		\end{aligned}
		\]
		It is easy to see that, $S_{E_{f}} = E_{f}E_{f}^T = \sum_{s=1}^{N}e_{k+s-1}e_{k+s-1}^T$ and $S_{U_{f}} = U_{f}U_{f}^T = \sum_{r=1}^{N}u_{k+r-1}u_{k+r-1}^T$
					\end{proof}
				
		Similarly, we can define $\bar{E}_{f}\bar{Y}_{p}^{T}$ as:
		\begin{equation}\label{eq_YpYpbar}
			\bar{E}_{f}\bar{Y}_{p}^{T} = S_{E_{f}}^{q}\bigg(\sum_{i=1}^{N}\sum_{j=1}^{N}\beta_{ij}e_{k+i-1}y_{j-1}^{T}\bigg)S_{Y_{p}}^{q}
		\end{equation}
	where $S_{Y_{p}} := Y_{p}Y_{p}^{T} = \sum_{r=1}^{N}y_{r-1}y_{r-1}^T$.
			
		Now, we show $\bar{E}_{f}$ remains approximately uncorrelated with compressed
		data $\bar{W}_{p}$ and $\bar{U}_{f}$. 
		\begin{claim}\label{Claim_Efbar}
			$\bar{E}_{f}/\begin{bmatrix}\bar{W}_{p}\\\bar{U}_{f}\end{bmatrix}\approx0$. 
		\end{claim}
	\begin{proof}
		Using definition \ref{def:stochasticorthogonalproj}, we evaluate only the first term of \eqref{expectation_sto} i.e. $\begin{bmatrix}\bar{E}_{f}\bar{W}^{T}_{p} & \bar{E}_{f}\bar{U}^{T}_{f} \end{bmatrix} = \begin{bmatrix}\bar{E}_{f}\bar{U}^{T}_{p} &\bar{E}_{f}\bar{Y}^{T}_{p} & \bar{E}_{f}\bar{U}^{T}_{f} \end{bmatrix}$. Now, using \eqref{eq_EfUfbar} and \eqref{eq_BBT_expectation},
		\begin{equation}\begin{aligned}
			\bar{E}_{f}\bar{U}^{T}_{f} \approx S_{E_{f}}^{q}\bigg(\sum_{i=1}^{N}e_{k+i-1}u_{k+i-1}^{T}\bigg) S_{U_{f}}^{q} \approx 0
		\end{aligned}
		\end{equation}	
		The last equality follows from the fact that $e$ and $u$ are uncorrelated i.e. $\mathbb{E}\big[e_{i}u_{j}^{T}\big] =0$ (see Assumption \ref{assumption:system}).
		Now using same argument as above for $\bar{U}^{T}_{p}$, we get $\bar{E}_{f}\bar{U}_{p}^{T} \approx 0$. Next we show that $\bar{E}_{f}\bar{Y}_{p}^{T} \approx 0$. Using \eqref{eq_YpYpbar} and \eqref{eq_BBT_expectation}:
			\[
		\begin{aligned}\bar{E}_{f}\bar{Y}_{p}^{T} \approx S_{E_{f}}^{q}\bigg(\sum_{i=1}^{N}e_{k+i-1}y_{i-1}^{T}\bigg)S_{Y_{p}}^{q}  \approx 0
		\end{aligned}
		\]
		The last equality is due to the fact that the output gets multiplied with future noise. Since they are uncorrelated,  i.e. $\mathbb{E}\big[e_{k+i}y_{i}^{T}\big] =0$, we assume that the middle term is approximately zero. Therefore, $\bar{E}_{f}\bar{W}^{T}_{p} \approx 0$ and $\bar{E}_{f}\bar{U}^{T}_{f} \approx 0$.
	\end{proof}
	
	Now, the orthogonal projection of $\bar{Y}_{f}$ onto the joint span of $\bar{W}_{p}$ and $\bar{U}_{f}$ is 
	\begin{equation}\label{eq_orthoproj_compressed}
		\begin{aligned}\bar{Y}_{f}/\begin{bmatrix}\bar{W}_{p}\\
				\bar{U}_{f}
			\end{bmatrix} & =\bar{L}_{p}\bar{W}_{p}/\begin{bmatrix}\bar{W}_{p}\\
				\bar{U}_{f}
			\end{bmatrix}+\Psi_{k}\bar{U}_{f}/\begin{bmatrix}\bar{W}_{p}\\
				\bar{U}_{f}
			\end{bmatrix}\\
			& +\Phi_{k}\underbrace{\bar{E}_{f}/\begin{bmatrix}\bar{W}_{p}\\
					\bar{U}_{f}
			\end{bmatrix}}_{\approx 0}\\
			& \approx \bar{L}_{p}\bar{W}_{p}+\Psi_{k}\bar{U}_{f}
		\end{aligned}
	\end{equation}
	The third term in the above equation becomes negligible using Claim \ref{Claim_Efbar}.
	The above result ensures that we can use projection to extract the desired subspace under assumption \ref{assumption:system} (similar to conventional method). 
	Recall the oblique projection $\zeta:=Y_{f}/_{U_{f}}W_{p}\in\mathbb{R}^{kp\times N}$
	and the compressed version $\bar{\zeta}:=\bar{Y}_{f}/_{\bar{U}_{f}}\bar{W}_{p}\in\mathbb{R}^{kp\times N_{c}}$.	Then the following result holds.
	\begin{lemma}
		\label{Lemma_obliqueproj} $\mathcal{R}(\bar{\zeta})=\mathcal{R}(\zeta)$
		a.s. 
	\end{lemma}
	\begin{proof}
		
		From  \eqref{eq_orthoproj_compressed}, $\bar{\zeta}:=\bar{Y}_{f}/_{\bar{U}_{f}}\bar{W}_{p}=\bar{L}_{p}\bar{W}_{p}$. The result follows from the fact that $\mathcal{R}(W_{p}) = \mathcal{R}(\bar{W}_{p})$ a.s. from Theorem \ref{Theorem_image_presevation}.
	\end{proof}
 We have shown that third term of \eqref{Ybarfinal} is negligible hence an equation  similar to \eqref{Yf_ortho_data} can be obtained for the compressed case using QR decomposition of $\bar{H}^{T}$:
	\begin{equation}\label{Yfbar_ortho_data}
		\bar{Y}_{f}/\begin{bmatrix}\bar{W}_{p} \\\bar{U}_{f}\end{bmatrix} = \bar{R}_{32}\bar{R}_{22}^{\dagger}\bar{W}_{p} + (\bar{R}_{31}-\bar{R}_{32}\bar{R}^{\dagger}_{22}\bar{R}_{21})\bar{R}^{-1}_{11}\bar{U}_f
	\end{equation}
	where $\bar{(.)}$ denote the equivalent matrices for the compressed case. In the above equation $\bar{R}_{11}$ is invertible due to the fact that $rank(\bar{U}_{f})$ = $rank(U_{f})$ a.s.

	Therefore, various projections of $\bar{Y}_{f}$ onto $\bar{U}_{f}$ and $\bar{W}_{p}$ as shown in \eqref{Ybarfinal} can be calculated from the compressed QR factors ($\bar{R}_{ij}$ $\forall i,j \in \{1, 2, 3\}$). Using the above Lemma and comparing \eqref{Yfbar_ortho_data} and \eqref{eq_orthoproj_compressed}, we get
	\begin{equation}\label{eq_zetabar}
		\bar{\zeta} = \underbrace{\bar{R}_{32}\bar{R}_{22}^{\dagger}}_{\bar{L}_p}\bar{W}_{p}
	\end{equation}
	\subsection{Projection: SVD Step}\label{subsection:SVDstep}
	As mentioned in preliminaries, after computing the oblique projection $\zeta \in \mathbb{R}^{kp \times N}$ (see \eqref{eqzeta}), the next step is to perform SVD according to \eqref{eqzetasvd}.  Since, we are interested to compute $\Theta_{k}$ using only the left singular vectors of $\zeta$, we show that an equivalent operation can be performed on $\bar{\zeta}$.
	\begin{theorem}
		\label{Theorem_thetak} There exists a decomposition of $\bar{\zeta}=\bar{L}_{p}\bar{W}_{p}=\bar{\Theta}_{k}X\in\mathbb{R}^{kp\times N_{c}}$
		such that $\mathcal{R}(\bar{\Theta}_{k})=\mathcal{R}(\Theta_{k})$
		a.s. where $\bar{\Theta}_{k}\in\mathbb{R}^{kp\times n}$
		and $X\in\mathbb{R}^{n\times N_{c}}$.
	\end{theorem}
	\begin{proof}
		 There exists a decomposition $\bar{\zeta}=\bar{\Theta}_{k}X$
		such that $rank(\bar{\zeta})=rank(\bar{\Theta}_{k})$ a.s.
		Moreover, a.s.,
		\[
		\begin{aligned} & \mathcal{R}(\bar{\zeta})=\mathcal{R}(\bar{\Theta}_{k})\hspace{0.5cm}(\text{from}\hspace{0.15cm}\bar{\zeta}\hspace{0.15cm}\text{decomposition})\\
			& \mathcal{R}({\zeta})=\mathcal{R}(\bar{\zeta})\hspace{0.5cm}\text{(from Lemma \ref{Lemma_obliqueproj})}\\
			& \mathcal{R}({\zeta})=\mathcal{R}({\Theta}_{k})\hspace{0.5cm}\text{(from Proposition \ref{Lemma_zetathetak})}
		\end{aligned}
		\]
		Using the above equations, $\mathcal{R}(\bar{\Theta}_{k})=\mathcal{R}(\Theta_{k})$
		a.s. 
	\end{proof}
	As mentioned before, we need only the left singular vectors of $\bar{\zeta}$ for further computation. Hence we perform reduced QR on $\bar{\zeta}^{T}$ while ignoring the Q-factors. Let, the QR decomposition of $\bar{\zeta}^{T}$ can be represented as,
	\begin{equation}\label{zetatranspose}
		\bar{\zeta}^T = \bar{Q}_{\zeta} \bar{R}_{\zeta}
	\end{equation}
	where $ \bar{R}_{\zeta} \in \mathbb{R}^{kp \times kp}$. Now, the SVD of $ \bar{R}_{\zeta}$ is computed as follows:
	\begin{equation}\label{Rzeta}\begin{aligned}
			 \bar{R}_{\zeta} = U_{r} \Sigma_{r} V_{r}^T &= \begin{bmatrix} U_1 & U_2\end{bmatrix} \begin{bmatrix} \Sigma_1 & 0 \\ 0 & 0\end{bmatrix} \begin{bmatrix} V_{1}^T \\ V_{2}^T\end{bmatrix}\\
			&= U_{1} \Sigma_{1} V_{1}^T
	\end{aligned}\end{equation} 
	
	\begin{lemma}\label{Lemma_LSV} 
		The left singular vectors of  $\bar{\zeta}$ are equal to the right singular vectors of $ \bar{R}_{\zeta}$. The singular values of $\bar{\zeta}$ and $ \bar{R}_{\zeta}$ are same. 
	\end{lemma}
	\begin{proof}
		Using, \eqref{zetatranspose} and \eqref{Rzeta},
		\begin{equation}\label{zetatransposefinal}\begin{aligned}
				 \bar{\zeta} = \bar{R}_{\zeta}^T \bar{Q}_{\zeta}^T\ &= V_{1}\Sigma_{1}U_{1}^TQ_{\zeta}^T= V_{1}\Sigma_{1}(U_{1}^TQ_{\zeta}^T) \\
				&= V_{1}\Sigma_{1}\bar{U}_{1}^T 
			\end{aligned}
		\end{equation}
		The last step follows from the fact that product of orthogonal matrices is an orthogonal matrix.
	\end{proof}	
	So, instead doing SVD of $\bar{\zeta}$ we perform SVD of $ \bar{R}_{\zeta}$ to estimate $\Theta_{k}$ which can be computed using \eqref{Rzeta} as,
	\begin{equation}\label{obsmatfinal}
		\hat{\Theta}_{k} = V_{1} \Sigma_{1}^{1/2}
	\end{equation}
	This step reduces the computation-cost as we use $\bar{\zeta} \in \mathbb{R}^{kp \times N_{c}}$ to compute $\Theta_{k}$ instead of $\zeta \in \mathbb{R}^{kp \times N}$.  
	
	\begin{theorem}
		\label{Lemma_structurepreserve} $\bar{Y}_{f}=\bar{\Theta}_{k}X+\Psi_{k}\bar{U}_{f}+\Phi_{k}\bar{E}_{f}$ 
	\end{theorem}
	\begin{proof}
		Using equation \eqref{eqYfinalOmega} and $\bar{\zeta}=\bar{\Theta}_{k}X$
		(see Theorem \ref{Theorem_thetak}), we can write $\bar{Y}_{f}$ as,
		\[
		\begin{aligned}\label{eqYbarfinal}\bar{Y}_{f} & =\bar{L}_{p}\bar{W}_{p}+\Psi_{k}\bar{U}_{f}+\Phi_{k}\bar{E}_{f}\\
			& =\bar{\Theta}_{k}X+\Psi_{k}\bar{U}_{f}+\Phi_{k}\bar{E}_{f}
		\end{aligned}
		\]
	\end{proof}
	The above Theorem proves that even after data compression we can
	use $\bar{\Theta}_{k}$, $\Psi_{k}$ and $\Phi_{k}$ to estimate the system parameters $\{A,B,C,D,K\}$ upto similarity transform.
	
	\subsection{Estimating model order, A and C}\label{subsection:AC}
	The model order ($n$) can be estimated from rank($R_{\zeta}$). In other words, the estimated rank is equal to the number of significant singular values of $R_{\zeta}$. This is usually inferred from a plot of the logs of the singular values in $\Sigma_{r}$ as obtained in \eqref{Rzeta}. Then, $A$ and $C$ are estimated using $\hat{\Theta}_{k}$ obtained in \eqref{obsmatfinal}, as in \eqref{eqAmatrix}. 
	
	\subsection{Estimating B and D}\label{subsection:BD_proposed}
	We propose a novel method, which does not depend on I/O data size ($N$), to estimate $B$ and $D$. Firstly, $\Psi_{k}$ is estimated using the already calculated LQ decomposition by comparing the second term of \eqref{Yfbar_ortho_data} and \eqref{eq_orthoproj_compressed}, 
	\begin{equation}\label{Psik}
		\hat{\Psi}_{k} = (\bar{R}_{31}-\bar{R}_{32}\bar{R}^{\dagger}_{22}\bar{R}_{21})\bar{R}^{-1}_{11}
	\end{equation}
	Then, $\{B,D\}$ can be estimated using $\hat{\Psi}_{k}$ and $\hat{\Theta}_{k}$ obtained from \eqref{obsmatfinal}. We extract the first $m$-columns from the estimated $\Psi_{k}$ calculated in \eqref{Psik} and compare with the first block column of $\Psi_{k}$ (see section \ref{section_conv} for structure of $\Psi_{k}$). Therefore,
	\begin{equation}\label{m_bd}
		\hat{M} := \underbrace{\hat{\Psi}_{k}(:,1:m)}_{\underset{\text{Data}}{\text{from \vspace{0.25cm}} \eqref{Psik}}} = \underbrace{\begin{bmatrix}
				D \\ CB\\ CAB\\\vdots\\CA^{k-2}B
		\end{bmatrix}}_{structure \rm\ of \rm\ \Psi_{k}}
	\end{equation}
	\begin{equation}\begin{aligned}\label{mupbdestimate}
			\hat{M}^{\uparrow} := \hat{M}((p+1):kp,:) &= \begin{bmatrix}
				CB\\ CAB\\\vdots\\CA^{k-2}B
			\end{bmatrix} = \begin{bmatrix}C\\CA\\\vdots\\CA^{k-2}\end{bmatrix} B\\
			&= \hat{\Theta}_{k}^{\downarrow}B
	\end{aligned}\end{equation}
	Next, $D$ and  $B$ can be estimated using \eqref{m_bd} and \eqref{mupbdestimate} respectively as, 
	\begin{equation}\label{eqDmatrix}
		\hat{D}= \hat{M}(1:p,:)
	\end{equation}  
	and
	\begin{equation}\label{eqBmatrix}\begin{aligned}
			\hat{B} &= arg \underset{B \in \mathbb{R}^{n \times m}}{min} || \hat{\Theta}^{\downarrow}_{k}B - \hat{M}^{\uparrow}||_{F}\\
			&= (\hat{\Theta}_{k}^{\downarrow})^{\dagger}\hat{M}^{\uparrow}
		\end{aligned}
	\end{equation} 
	We can see that $D$ is estimated directly by reading the $1^{st}$-$p$ rows of $\hat{M}$. Now, to estimate $B$, negligible computation is required. Since, $(\hat{\Theta}_{k}^{\downarrow})^{\dagger}$ has been already computed to estimate $A$ (see  \eqref{eqAmatrix}) and $\hat{M}^{\uparrow}$ is just the shifted version of $\hat{M}$ (see \eqref{mupbdestimate}), hence, the only computation required to estimate $B$ is the multiplication of two small matrices in \eqref{eqBmatrix}. 
	
	\subsection{Estimation of K}\label{subsection:section_stochasticident}
	We propose a new fast method, which is independent of data-size $N$, to estimate $K$. We use an approach similar to the $\{B, D\}$ estimation. Here we exploit the structure of $\Phi_{k}$ (see section \ref{section_conv}), thereby avoiding computation on the full sized data matrices. 
	From LQ decomposition (see \eqref{Hmatrix}),  $Y_{f}$ can be written as,
	\begin{equation}\label{eqYfinal}
		Y_{f} = R_{31}Q_{1}^{T} + R_{32}Q_{2}^{T} + R_{33}Q_{3}^{T}
	\end{equation}
	Similarly for the randomized case, $\bar{Y}_{f}$ can be written as,
	\begin{equation}\label{eqYbfinal}
		\bar{Y}_{f} = \bar{R}_{31}\bar{Q}_{1}^{T} + \bar{R}_{32}\bar{Q}_{2}^{T} + \bar{R}_{33}\bar{Q}_{3}^{T}
	\end{equation}
	For the uncompressed case as shown in \cite{Katayama_LQfullmodel}, the stochastic component can be obtained  using \eqref{Yfuturesto} and \eqref{eqYfinal},
	\begin{equation}\label{R33}
		R_{33}Q^T_{3} =  \Phi_{k} E_f
	\end{equation}
	We show that even after data
	compression, the third term of \eqref{eqYbfinal} obeys a similar equality as \eqref{R33}.
	\begin{lemma}\label{Lemma_phikbarEf}
		$\Phi_{k}\bar{E}_{f}\approx\bar{R}_{33}\bar{Q}_{3}^{T}$
	\end{lemma}
	\begin{proof}
		From Claim \ref{Claim_Efbar}, $\bar{E}_{f}/\begin{bmatrix}\bar{W}_{p}\\
			\bar{U}_{f}
		\end{bmatrix}\approx 0$ which in turn implies that $\bar{Y}_{f}/\begin{bmatrix}\bar{W}_{p}\\
			\bar{U}_{f}
		\end{bmatrix}^{\perp}\approx \Phi_{k}\bar{E}_{f}$ (see \eqref{eqYfinalOmega} and \eqref{Ybarfinal}). From
		the LQ decomposition in \eqref{eqYbfinal}, $\bar{Y}_{f}/\begin{bmatrix}\bar{W}_{p}\\
			\bar{U}_{f}
		\end{bmatrix}^{\perp}=R_{33}\bar{Q}_{3}^{T}$. Equating both we get $\Phi_{k}\bar{E}_{f}\approx \bar{R}_{33}\bar{Q}_{3}^{T}$.
	\end{proof}
	Next we are interested in exploiting the structure of $\bar{R}_{33}$ for extracting the stochastic component. For that we analyze the structural properties of $\bar{E}_{f}\bar{E}_{f}^T$. 
	\begin{lemma}\label{Lemma_rank_barEf}
		$rank(\bar{E}_{f}\bar{E}_{f}^{T}) = rank(E_{f}E_{f}^{T})$ a.s.
	\end{lemma}
	\begin{proof}
		Using Lemma \ref{Theorem_rank_preservation}, $rank(\bar{E}_{f}) = rank(E_{f})$ a.s. Now,
		\begin{equation*}\begin{aligned}
				& rank(\bar{E}_{f}\bar{E}_{f}^{T}) = rank(\bar{E}_{f})\\
				& rank(\bar{E}_{f}) = rank(E_{f})\\
				& rank(E_{f}) = rank(E_{f}E_{f}^{T})\\
			\end{aligned}
		\end{equation*}
	\end{proof}
	Since it is easy to see that $E_{f}E_{f}^{T}$ is positive definite, Lemma \ref{Lemma_rank_barEf} implies $\bar{E}_{f}\bar{E}_{f}^{T}$ is also positive definite a.s.
 
	\begin{lemma}\label{Lemma_barEf}
		$ \bar{E}_{f}\bar{E}_{f}^{T} $ is a block-diagonal matrix.
	\end{lemma}
	\begin{proof}
		We know, 
		\begin{equation*}
			\begin{aligned}\bar{E}_{f} &= ({E}_{f}{E}_{f}^{T})^{q}{E}_{f}\mathcal{C} = S_{E_{f}}^{q}\bigg(\sum_{i=1}^{N} e_{k+i-1}\alpha_{k+i-1}^{T}\bigg)
			\end{aligned}
		\end{equation*}
Then $\bar{E}_{f}\bar{E}_{f}^{T}$ can be written as:
		\begin{equation}\label{eq_EfbarEfbar}
			\begin{aligned}
				\bar{E}_{f}\bar{E}_{f}^{T} 
				& =S_{E_{f}}^{q}\bigg(\sum_{i=1}^{N}\sum_{j=1}^{N}\beta_{ij}e_{k+i-1}e_{k+j-1}^{T}\bigg)S_{E_{f}}^{q}\\
			\end{aligned}
		\end{equation}
			Define, $e_{k+i-1} :=
			\big[ e(k+i-1)^{T} \hspace{0.2cm} e(k+i)^{T} \hspace{0.2cm} \hdots \hspace{0.2cm} e(2k+i-2)^{T} \big]^{T} 
		 \in \mathbb{R}^{kp}$ with $e(l) \in \mathbb{R}^{p}$. Then $S_{E_{f}}$ in matrix from:
  \small
  		\begin{equation*}\begin{aligned}
					&S_{E_{f}} = E_{f}E_{f}^{T} = \sum_{i=1}^{N}e_{k+i-1}e_{k+i-1}^T\\  &\approx \text{diag}\Big(\sum_{i=1}^{N}e(k+i-1)e(k+i-1)^{T},\sum_{i=1}^{N}e(k+i)e(k+i)^{T},\\
      &\hspace{1.2in}\hdots,\sum_{i=1}^{N}e(2k+i-2)e(2k+i-2)^{T}\Big)
					 \\
    &\text{The cross-terms are neglected since they are products of }\\
    & \text{uncorrelated variables.}\\
				& \text{(multiply and divide by $N$)}\\
					& \approx N \text{diag}\Big(
					\mathbb{E}[(e(k+i-1)e(k+i-1)^{T})],\mathbb{E}[(e(k+i)e(k+i)^{T})]\\ &\hspace{1.2in}\hdots,\mathbb{E}[(e(2k+i-2)e(2k+i-2)^{T})]\Big)
				\\
				&=N\begin{bmatrix}
					\eta& 0 & \hdots & 0\\0&\eta& \hdots&0\\
					\vdots & \vdots & \ddots & \vdots\\
					0&0&\hdots& \eta
				\end{bmatrix} =: N \Omega
			\end{aligned}
		\end{equation*}
\normalsize
where, $\eta >0 \in \mathbb{R}^{p \times p}$ is the noise/innovation covariance matrix. Similarly, the middle term of \eqref{eq_EfbarEfbar} can be approximated using \eqref{eq_BBT_expectation} as:
		\begin{equation*}\begin{aligned}
							\sum_{i=1}^{N}\sum_{j=1}^{N}\beta_{ij}e_{k+i-1}e_{k+j-1}^{T}  & \approx \sum_{i=1}^{N}e_{k+i-1}e_{k+i-1}^{T}
			\end{aligned}
		\end{equation*}
		Therefore:
		\begin{equation}\label{eq_EfEf_final}
			\bar{E}_{f}\bar{E}_{f}^{T} \approx N^{(2q+1)}\Omega^{(2q+1)}
		\end{equation}
	\end{proof}
	\begin{lemma}\label{Lemma_R33pd}
		$\bar{R}_{33}\bar{R}^T_{33}$ is a symmetric positive definite matrix.
	\end{lemma}
	\begin{proof}
		Using Lemma \ref{Lemma_phikbarEf},
		\begin{equation*}\label{R33_symm}\begin{aligned}
				\bar{R}_{33}\bar{R}^T_{33} &= (\Phi_{k} \bar{E}_f \bar{Q}_{3}) (\Phi_{k} \bar{E}_f \bar{Q}_{3})^T\\
				&= \Phi_{k} \bar{E}_f \bar{E}^T_f\Phi^T_{k}
			\end{aligned}
		\end{equation*}
		
		Now, multiply by $1/N^{(2q+1)}$ on both sides, we get
		\begin{equation}\label{R33_cov}\begin{aligned}
				\frac{1}{N^{(2q+1)}}	\bar{R}_{33}\bar{R}^T_{33} &=\frac{1}{N^{(2q+1)}}\Phi_{k} \bar{E}_f \bar{E}_f^T\Phi^T_{k}\\
				&= \Phi_{k} \bigg\{\frac{1}{N^{(2q+1)}}\bar{E}_f \bar{E}_f^T\bigg\} \Phi^T_{k}\\
				&=  \Phi_{k}\Omega^{(2q+1)} \Phi^T_{k} \hspace{0.5cm} (\text{from} \hspace{0.2cm} \eqref{eq_EfEf_final})
			\end{aligned}
		\end{equation}
		Since, $\Omega^{(2q+1)} >0$ and is block diagonal, the Cholesky factors of $\Omega^{(2q+1)}$, denoted by $\Omega^{(2q+1)/2}$, is also block diagonal. Hence, the above expression can be factorized as $(\Phi_{k}\Omega^{(2q+1)/2)})(\Phi_{k}\Omega^{(2q+1)/2})^{T}$. From the structures of $\Phi_{k}$ and $\Omega^{(2q+1)/2}$, it follows that their product is block lower triangular with each of the diagonal blocks being positive definite. Hence the product is  also positive definite.
	\end{proof}
	We can use the above Lemma to estimate the stochastic component. To achieve that we make use of Cholesky decomposition.
	From \eqref{R33_cov}, $\frac{1}{N^{(2q+1)}}\bar{R}_{33}\bar{R}_{33}^{T} =  \Phi_{k}\Omega^{(2q+1)} \Phi^T_{k}$. Then the Cholesky factor becomes,
	\begin{equation}\label{R33_tilde}
		\frac{1}{\sqrt{N^{(2q+1)}}} \bar{R}_{33} = \Phi_{k}\begin{bmatrix}
			\omega& 0 & \hdots & 0\\0&\omega& \hdots&0\\
			\vdots & \vdots & \ddots & \vdots\\
			0&0&\hdots& \omega
		\end{bmatrix}  
	\end{equation}
	where $\omega\omega^T = \eta^{(2q+1)}$ and $\omega \in \mathbb{R}^{p \times p}$. By defining, $\tau := \frac{1}{\sqrt{N^{(2q+1)}}} \in \mathbb{R}$ and substituting the structure of $\Phi_k$ above we get:
	\begin{equation}\begin{aligned}\label{structure_NOISE}
			\tau\bar{R}_{33}&= \begin{bmatrix}
				I & 0 & \hdots & 0\\CK & I & \ddots & \vdots\\ \vdots & \ddots & \ddots & 0 \\ CA^{k-2}K & \hdots & CK & I  \end{bmatrix} \begin{bmatrix}
				\omega& 0 & \hdots & 0\\0&\omega& \hdots&0\\
				\vdots & \vdots & \ddots & \vdots\\
				0&0&\hdots& \omega
			\end{bmatrix} \\
			&= \begin{bmatrix}
				\omega & 0 & \hdots & 0\\CK\omega & \omega & \hdots&0\\ \vdots & \vdots & \ddots& \vdots \\ CA^{k-2}K\omega & CA^{k-3}\omega & \hdots & \omega \end{bmatrix}
		\end{aligned} 
	\end{equation}
	Now we exploit the structure of the above equation to estimate $K$. Define, 
	\begin{equation}\label{Pk}
		\hat{P}_k := \underbrace{\tau\bar{R}_{33}(:,1:p)}_{\underset{\text{Data}}{\text{from LQ \vspace{0.2cm}} \eqref{eqYbfinal}}} = \underbrace{\begin{bmatrix}
				\omega \\ CK\omega\\ \vdots \\ CA^{k-2}\omega
		\end{bmatrix}}_{structure \rm\ of \rm\ \eqref{structure_NOISE}}
	\end{equation}
	Hence, $K$ can be estimated by defining $\hat{P}_{k}^{\uparrow} := \hat{P}_{k}(p+1:end,:)$ and noting that: 
	\begin{equation*}
		\hat{P}_{k}^{\uparrow} = \begin{bmatrix}
			C \\ CA\\ \vdots \\ CA^{k-2}
		\end{bmatrix} K\omega = \hat{\Theta}^{\downarrow}_{k} K\omega
	\end{equation*}
	Therefore,
	\begin{equation}\label{Kestimate}\begin{aligned}
			\hat{K} &=arg \underset{K \in \mathbb{R}^{n \times p}}{min} || \hat{\Theta}^{\downarrow}_{k}K - \hat{P}_{k}^{\uparrow}\hat{\omega}^{-1}||_{F}\\
			&= (\hat{\Theta}^{\downarrow}_{k})^{\dagger}\hat{P}_{k}^{\uparrow}\hat{\omega}^{-1}
		\end{aligned}
	\end{equation}
	where, $\omega$ can be estimated as $\hat{\omega} = \hat{P}_k(1:p,:)$ and is invertible.  Since, $(\hat{\Theta}^{\downarrow}_{k})^{\dagger}$ has been already calculated to estimate $A$ (see \eqref{eqAmatrix}), the only computation left at this step is calculating the inverse of $\hat{\omega} \in \mathbb{R}^{p \times p}$ and the matrix multiplication required in \eqref{Kestimate}.
	\begin{algorithm}\label{algo:estimateK}
	Obtain $\bar{R}_{33}$ from LQ decomposition (ref \eqref{eqYbfinal})\\
	Calculate $\tau\bar{R}_{33}$ and $\hat{P}_k := \tau\bar{R}_{33}(:,1:p)$\\
	Estimate $\omega$ as $\hat{\omega} = \hat{P}_k(1:p,:)$\\
	Define, $\hat{P}_{k}^{\uparrow} := \hat{P}_{k}(p+1:end,:)$\\
	Estimate $K$ using \eqref{Kestimate}
	\caption{Estimate of K}
\end{algorithm}
	 Hence, the estimate of K is obtained using $\bar{R}_{33} \in \mathbb{R}^{kp \times kp}$ and the multiplication of small matrices (see Algorithm \ref{algo:estimateK}).
	
	In summary, we have estimated the combined dynamics of the system up to similarity transform using equations \eqref{eqAmatrix}, \eqref{eqDmatrix}, \eqref{eqBmatrix} and \eqref{Kestimate}.

	\subsection{Proposed Alogrithm steps} The input-output training data i.e. $\{u(i)\}^{N_{t}-1}_{i=0}$ and $\{y(i)\}^{N_{t}-1}_{i=0}$ is given, where, $N_{t} := 2k+N-2$ (number of training samples). Also, we assume that a input-output validation dataset is provided, having $N_{v}$ samples which are unseen by the estimated model. To simplify the proposed algorithm, we divide it in two-parts, namely $(i)$ Streaming data compression (SDC) (see Algorithm \ref{algo:fastdatapre}) and, $(ii)$ Main algorithm. 
	The steps of the proposed main algorithm are summarized as follows:
	\begin{algorithm}[h]\label{algo:fastsubspaceid}
		\textbf{Input:} $\bar{U}_{p} \in \mathbb{R}^{km \times N_{c}}$, $\bar{U}_{f}\in \mathbb{R}^{km \times N_{c}}$, $\bar{Y}_{p}\in \mathbb{R}^{kp \times N_{c}}$ and $\bar{Y}_{f}\in \mathbb{R}^{kp \times N_{c}}$. (Output of Algorithm \ref{algo:fastdatapre} i.e. SDC)\\
		Perform reduced QR algorithm on $\bar{H}^{T}\in\mathbb{R}^{N_{c}\times2k(m+p)}$ to obtain $\bar{R}_{d}$.\\
		Extract $\bar{L}_{p}$ and $\Psi_{k}$ from $\bar{R}$ obtained above using \eqref{eq_zetabar} and \eqref{Psik} respectively. \\
		Perform reduced QR algorithm using $\bar{L}_{p}$ and $\bar{W}_{p}$ to obtain $\bar{R}_{\zeta}$ \eqref{Rzeta}.\\
		Perform SVD on $\bar{R}_{\zeta} \in \mathbb{R}^{kp \times kp}$ to estimate $\Theta_{k}$ \eqref{obsmatfinal}.\\
		Estimate A and C using $\Theta_{k}$ and it's shifted version i.e. $\Theta_{k}^{\uparrow}$ and $\Theta_{k}^{\downarrow}$ as in \eqref{eqAmatrix}.\\ 
		Estimate B and D using $1^{st}$-m columns of $\Psi_{k}$ and $\Theta_{k}^{\downarrow}$ as in \eqref{eqBmatrix} and \eqref{eqDmatrix}.\\
		Estimate K using Algorithm \ref{algo:estimateK}.\\
		Validate model using validation dataset (unseen by estimated model) and MSE.\\
		\textbf{Output:} Estimated $\{A, B, C, D, K\}$.\\ 
		\caption{Fast randomized subspace system identification (FR2SID) algorithm}
	\end{algorithm}
	
	In the proposed algorithm, if the MSE on validation data set (see step 9 in Algorithm \ref{algo:fastsubspaceid}) is high then tuning of hyper-parameters like over-sampling parameter ($l$) and $k$ might be required. It is a well known fact that low SNR leads to increased MSE. For such cases, choosing a high value of $l$ tend to give better results (\cite{anderson2022_ESRSVD}, \cite{Zhang_perturbsvd_2022}). We also present a summary or mind-map of the proposed algorithm in Fig. \ref{fig:Summary}. The blue dotted portion is the proposed algorithm.
	
	\begin{figure}[htbp]
		\begin{adjustbox}{max height=0.43\textheight,center}
			\begin{tikzpicture}[node distance=3cm]
				
				\node (in1) [io] {Load Input-Output data};
				\node (dec1) [decision, below of=in1, yshift=-0.5cm] {\begin{tabular}{ccc}
						Does $H$ fits into \\ slow memory?\\(see Note \ref{Case:RAM_space})
				\end{tabular}};
				\node (pro1) [process, right of=dec1, xshift=2cm] {Out-of-memory. \\Need Streaming Algorithm};
				\node (pro1a) [process, below of=pro1, xshift = 0cm, yshift = 0.5cm] {Randomized Subspace approach};
				\node (pro1b) [process, below of=pro1a,xshift = 0cm, yshift = 0.75cm] {Randomized Range Approx.- Streaming data\\ (Algorithm \ref{algo:fastdatapre})};
				\node (pro1c) [process, below of =pro1b, xshift = 0cm, yshift = 1cm] {FR2SID \\ (Algorithm \ref{algo:estimateK} and \ref{algo:fastsubspaceid})};
				\node (pro2) [process, left of=dec1, xshift=1.5cm, yshift = -5.5cm] {Conventional Subspace Algorithm\\(Algorithm \ref{algo:conv_SID})};
				
				\node[fit=(pro1a)(pro1b)(pro1c),myfit] (myfit1) {};
				\node[mytitle] at (myfit1) {\begin{tabular}{cc}
						Proposed Algorithm
				\end{tabular}};

				\draw [arrow] (in1) -- (dec1);
				\draw [arrow] (dec1) -- node[anchor=south] {No} (pro1);
				\draw [arrow] (pro1) -- (pro1a);
				\draw [arrow] (pro1b) -- (pro1c);
				\draw [arrow] (pro1a) -- (pro1b);
				\draw [arrow] (dec1.south) -- node[anchor=west] {Yes}(pro1b.west);
				\draw [arrow] (dec1.south) -- node[anchor=east] {Yes}(pro2.north);
			\end{tikzpicture}
		\end{adjustbox}
		\caption{Summary of the proposed algorithm}
		\label{fig:Summary}
	\end{figure}
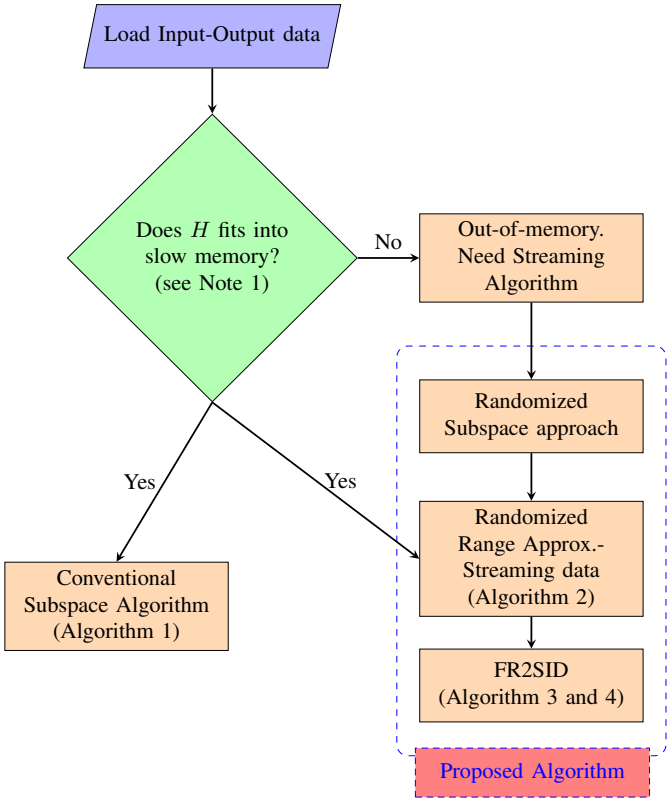
	
	\section{Algorithm performance}	\label{section:algo_performance}
	Several improvements in terms of the algorithm cost $C$ (see \eqref{Algo_cost}) are achieved in the proposed algorithm. All the calculations involving $q$-terms are derived for $q \in \{0,1\}$. 
	\subsection{SDC Analysis}\label{subsection:SDC_analysis}
	In this section we compute the memory-cost, flops and data movement for the proposed SDC algorithm (see  section \ref{subsection:SDC}). Let, $H_{i} := \begin{bmatrix}
		U_{f_{i}}^{T}& U_{p_{i}}^{T} &Y_{p_{i}}^{T}& Y_{f_{i}}^{T}\end{bmatrix}^{T} \in \mathbb{R}^{2k(m+p) \times N_{d}}$ and $\bar{H}_{i} := (H_{i}^TH_{i})^{q}H_{i}\mathcal{C}_{i} \in \mathbb{R}^{2k(m+p) \times N_{c}}$.
	\begin{itemize}
		\item Memory-cost ($M_{SDC}$): The space required to store $H_{i}$ is $2k(m+p)N_{d}$. Then $\mathcal{C}_{i}$ requires $N_{d}N_{c}$ words, and finally to store $\bar{H}$ we need $2k(m+p)N_{c}$ words. Adding all these and after simplification,
		\begin{equation*}\begin{aligned}
				M_{SDC} &= (4k(m+p)+l)N_{d} +  4k^2(m+p)^2\\& + 2kl(m+p)
			\end{aligned}
		\end{equation*}
		\item Flops ($F_{SDC}$): The flop-count can be calculated very easily considering the dimensions of the matrices $H \in \mathbb{R}^{2k(m+p) \times N}$ and $\mathcal{C} \in \mathbb{R}^{N \times N_{c}}$:
		\begin{equation}\label{FpSDC}\begin{aligned}
				F_{SDC} &= (q+1)8k^2(m+p)^2N + 4kl(m+p)N\\&+q(16k^3(m+p)^3+8k^2l(m+p)^2)	
			\end{aligned}
		\end{equation}
		\item Data-movement ($DM_{SDC}$):
		It is easy to see that in order to perform matrix multiplication we need to read both the matrices. The matrix $H$ is read $(q+1)$ times while $\mathcal{C}$ is read once; therefore $\#words_{read} = (q+2)2k(m+p)N + lN$. The number of words write due to $\bar{H}$ is $\#words_{write} = 4k^2(m+p)^2 +2kl(m+p)$. Now, in each iteration $(q+2)$ blocks/messages is read so $\#messages_{read} = (q+2)d$ while $\#messages_{write} = 1$ (since $\bar{H}$ fit into fast memory due to Assumption \ref{assumption:W_Nd}). Therefore,
		\begin{equation}\begin{aligned}\label{DMpSDC}
				DM_{SDC} &= (q+2)2k(m+p)N + lN + 4k^2(m+p)^2 \\&+ 2kl(m+p) + (q+2)d + \mathcal{O}(1)
			\end{aligned}
		\end{equation}
	\end{itemize}
	\begin{note}
		Due to Assumption \ref{assumption:W_Nd}, number of messages ($\#messages$) for all subsequent steps will be of $\mathcal{O}(1)$.
	\end{note}
	\subsection{Projection: QR Analysis}\label{subsection:QRanalysis}
	In this sub-section we compute the memory-cost, flops and data movement for the reduced QR (RQR) step (see section \ref{subsection:QRstep}).
	\begin{itemize}
		\item Memory-cost ($M_{RQR}$): The memory-cost is due to $\bar{R}$ only (since $\bar{H}$ is already considered in the previous step and $\bar{Q}$ is ignored). Therefore,
		\begin{equation*}\begin{aligned}
				M_{RQR} &= 4k^2(m+p)^2
			\end{aligned}
		\end{equation*}
		\item Flops ($F_{RQR}$): The number of flops required for RQR is 
		\begin{equation}\label{FpRQR}\begin{aligned}
				F_{RQR} &= \frac{32}{3}k^3(m+p)^3+8k^2l(m+p)^2
			\end{aligned}
		\end{equation}
		
		\item Data Movement ($DM_{RQR}$): The matrix $\bar{H}$ has to be read only once i.e. $\#words_{read_{RQR}} = 4k^2(m+p)^2 + 2kl(m+p)$. Next, we write only R-factors of size $2k(m+p) \times 2k(m+p)$ while we ignore Q-factors. Therefore, $\#words_{write_{RQR}} = 4k^2(m+p)^2$. Hence,
		\begin{equation}\begin{aligned}\label{DMpRQR}
				DM_{RQR} &= \#words_{RQR} +  \#messages_{RQR}\\
				&=  8k^2(m+p)^2 + 2kl(m+p) + \mathcal{O}(1)
			\end{aligned}
		\end{equation}
	\end{itemize}
	
	\subsection{Projection: SVD Analyis}\label{subsection:SVDanalysis}
	In this sub-section we compute the memory-cost, flops and data movement for SVD step (see section \ref{subsection:SVDstep}).
	\begin{itemize}
		\item Memory-cost ($M_{SVD}$): The memory-cost is due to $\bar{\zeta}$  and $\bar{R}_{\zeta}$ (since $\bar{L}_{p}$ and $\bar{W}_{p}$ are already considered in previous steps). Recalling the dimensions of these matrices:
		\begin{equation*}\begin{aligned}
				M_{SVD} &= 3k^2p^2+2k^2mp+klp
			\end{aligned}
		\end{equation*}
		\item Flops ($F_{SVD}$): The number of flops for SVD comprises of the flops associated with the three computation steps as shown in table \ref{tab:flopcount_svd}.
		\begin{table}[h]
			\setlength\extrarowheight{2pt}
			\begin{center}
				\caption{Flop-count for SVD step}
				\label{tab:flopcount_svd}
				\scalebox{1}{
					\begin{tabular}{|l|l|}
						\hline
						\multicolumn{1}{|c|}{\textbf{Algorithm steps}} & \multicolumn{1}{c|}{\textbf{Flop-count}}                  \\ \hline
						\begin{tabular}[c]{@{}l@{}}Matrix multiplication \\($\bar{\zeta} = \bar{L}_{p}\bar{W}_{p}$)    \end{tabular}          &   \begin{tabular}[c]{@{}l@{}}$ 4k^3pm^2+8k^3p^2m+4k^3p^3$+\\$2k^2lpm+2k^2lp^2$ \end{tabular}          \\ \hline
						QR on $\bar{\zeta}^{T} $                        &          \begin{tabular}[c]{@{}l@{}}$4k^3p^2m+2k^2lp^2+(10/3)k^3p^3$ \end{tabular}                                   \\ \hline
						SVD on $\bar{R}_{\zeta}$                              &          $4k^3p^3$                                \\ \hline
				\end{tabular}}
			\end{center}
		\end{table}
		Therefore, we can write: 
		\begin{equation}\label{FpSVD}\begin{aligned}
				F_{SVD} &= k^3(4pm^2+12p^2m+(34/3)p^3)\\
				&+k^2l(2pm+4p^2)
			\end{aligned}
		\end{equation}
		\item Data Movement ($DM_{SVD}$): It is easy to see that, the data movement will be of $\mathcal{O}(k^2+kl)$ due to size of matrices involved in computation. Therefore,
		\begin{equation}\begin{aligned} \label{DMpSVD}
				DM_{SVD} &= \mathcal{O}(k^2+kl)
			\end{aligned}
		\end{equation}
	\end{itemize}
	\begin{note}
		In summary,
		\begin{itemize}
			\item Flop-count post data compression and QR step is $\mathcal{O}(k^3+k^2l)$.
			\item Data movement post data compression step is of  $\mathcal{O}(k^2+kl)$.
		\end{itemize}
	\end{note}
	\subsection{Model parameters estimation analysis}\label{subsection:ABCDKanalysis}
	In this sub-section we compute the memory-cost, flops and data movement required to estimate $\{A, B, C, D, K\}$ (see section \ref{subsection:AC}, \ref{subsection:BD_proposed} and \ref{subsection:section_stochasticident}).
	
	\begin{itemize}
		\item Memory-cost ($M_{ABCDK}$) : From table \ref{tab:memorycost}, memory-cost is  $\mathcal{O}(k^2)$.
		\begin{table}[h]
			\setlength\extrarowheight{2pt}
			\begin{center}
				\caption{Memory-cost for $\{A, B, C, D, K\}$}
				\label{tab:memorycost}
				\scalebox{1}{
					\begin{tabular}{|l|l|}
						\hline
						\multicolumn{1}{|c|}{\textbf{Action}}                         & \multicolumn{1}{c|}{\textbf{Memory-cost}} \\ \hline
						Estimating \{A, B, C, D\}                                 & \begin{tabular}[c]{@{}l@{}}$(n+p)(m+n)+kp(3n+2m)$ \end{tabular}                   \\ \hline
						Estimating \{K\}                                 & $k^2p^2+2kp^2+np$                    \\ \hline
				\end{tabular}}
			\end{center}
		\end{table} 
		\item Flop-count ($F_{ABCDK}$): Detailed flop-count computation has been shown in table \ref{tab:ABCDK} and \ref{tab:flopcount}. It follows that $F_{ABCDK}$ for proposed method is of $\mathcal{O}(k^3)$ while for conventional methods it is of $\mathcal{O}(n^2N)$ (see Table \ref{tab:flopcount_conv}).
		\begin{table}[h]
			\setlength\extrarowheight{2pt}
			\begin{center}
				\caption{Flop-count to estimate $\{A, B, C, D, K\}$}
				\label{tab:ABCDK}
				\scalebox{1.1}{
					\begin{tabular}{|l|l|}
						\hline
						\multicolumn{1}{|c|}{\textbf{\begin{tabular}[c]{@{}l@{}}Algorithm\\ steps \end{tabular} }} & \multicolumn{1}{c|}{\textbf{Flop-count}}                  \\ \hline
						$\hat{A}$   &  $2np^2(k-1)^2+18n^{2}p(k-1)+8n^3$ \\ \hline
						$\hat{B}$   &  \begin{tabular}[c]{@{}l@{}}$\frac{2}{3}(km)^3+2k^2mp(m+p)+k^2pm+$\\$2k^3pm^2+2nmp^2(k-1)^2$\end{tabular}           \\ \hline
						$\hat{C}$   &  read-off from $\hat{\Theta}_{k}$           \\ \hline
						$\hat{D}$   &  read-off from $\hat{\Psi}_{k}$           \\ \hline
						$\hat{K}$   &   $\frac{2}{3}p^3+k^2p^2+2pn(p+k-1)$ \\ \hline
				\end{tabular}}
			\end{center}
		\end{table}
		\begin{table}[h]
			\setlength\extrarowheight{2pt}
			\begin{center}
				\caption{Flop-count for intermediate calculations}
				\label{tab:flopcount}
				\scalebox{1.1}{
					\begin{tabular}{|l|l|}
						\hline
						\multicolumn{1}{|c|}{\textbf{Algorithm steps}} & \multicolumn{1}{c|}{\textbf{Flop-count}}                  \\ \hline
						$\bar{R}^{\dagger}_{22}$                             & $8(k(m+p))^3$                                            \\ \hline
						$\bar{L}_{p}$                                  & $2k^3p(m+p)^2$                                            \\ \hline
						$\hat{\Psi}_{k}$                                     & \begin{tabular}[c]{@{}l@{}}$\frac{2}{3}(km)^3+2k^3mp(m+p)+$\\$k^2pm+2k^3pm^2$ \end{tabular}           \\ \hline
						$\hat{\Theta}_{k}$                                   & $2kpn^2+n$                                       \\ \hline
				\end{tabular}}
			\end{center}
		\end{table}
	\end{itemize}
	
	\subsection{Final performance} 
	In this section we show proposed algorithm reduces the cost $C$ (see \eqref{Algo_cost}) as compared to the conventional subspace identification algorithms. Let subscript $p$ denotes quantities for proposed algorithm while $conv$ denotes conventional algorithms. The computation in this subsection is based on section \ref{subsection:SDC_analysis}, \ref{subsection:QRanalysis}, \ref{subsection:SVDanalysis} and \ref{subsection:ABCDKanalysis}.
	In the following Lemma, we show that the memory cost for proposed algorithm is less than conventional algorithm due to iterative update, data compression and ignoring $Q$-factors in the proposed method.
	\begin{lemma}\label{Lemma_memorycost}
		$M_{p} < M_{conv}$.
	\end{lemma}
	\begin{proof}
		Adding all memory cost from previous subsections: 
		\begin{equation}\label{Mproposed}\begin{aligned}
				M_{p} &= M_{SDC} + M_{RQR} + M_{SVD} + M_{ABCDK}\\
				&=  (4k(m+p)+l)N_{d} +  4k^2(m+p)^2+ 2kl(m+p)\\
				&+4k^2(m+p)^2+ 3k^2p^2+2k^2mp+klp + M_{ABCDK}\\
				& \approx (4k(m+p)+l)N_{d}
		\end{aligned}\end{equation}
Since, $k << N_{d}$ implies $k^2 << kN_{d}$ so the terms containing $k^2$ is ignored in the above equation. Moreover, the oversampling parameter ($l$) is very small (fixed at $l=5$, see \cite{halko_2011}, \cite{anderson2022_ESRSVD}), hence it can also be ignored.
		Using \eqref{Mconv} and \eqref{Mproposed}, it is easy to deduce that  $M_{p} < M_{conv}$. 
	\end{proof}
	
 Next we show flop-count comparison.
	
	\begin{lemma}\label{Lemma_flopcount}
		Flop-count for proposed method is  $F_{p} \approx  (q+1)8k^2(m+p)^2N + 4kl(m+p)N$ and $F_{p} < F_{conv}$.
	\end{lemma}
	\begin{proof}
		Adding all flop-counts from \eqref{FpSDC}, \eqref{FpRQR}, \eqref{FpSVD} we get,
		\begin{equation}\label{Fproposed}
			\begin{aligned}
				&F_{p} = F_{SDC} + F_{RQR} + F_{SVD} + F_{ABCDK}\\
				&= (q+1)8k^2(m+p)^2N + 4kl(m+p)N\\&+k^3\bigg(16q(m+p)^3+4pm^2+12p^2m+\frac{34}{3}p^3+\frac{32}{3}(m+p)^3\bigg)\\
				&+k^2l\bigg(8(q+1)(m+p)^2+2pm+4p^2\bigg) + F_{ABCDK}\\
				& \approx  (q+1)8k^2(m+p)^2N + 4kl(m+p)N	
			\end{aligned}
		\end{equation}
		Since $k<<N$ therefore $k^3 << k^2N$. Hence we have only taken the dominating terms for the approximate calculation of flops in the last step. Therefore, on comparing \eqref{FconvABCDK} and \eqref{Fproposed}, $F_{p} < F_{conv}$.
	\end{proof}
	Although the proposed method contains $k^2N$ term corresponding to SDC algorithm but once data-compression is done, all other computations have a complexity of $\mathcal{O}(k^3)$. While for conventional method, not only the QR factorization, matrix multiplication ($L_{p}W_{p}$) and SVD steps are of $\mathcal{O}(k^2N)$, further estimation of model parameters requires $\mathcal{O}(n^2N)$ computations (see Table \ref{tab:Performance_comparison}). To emphasize this advantage we define the \% reduction in flop-count for estimating system parameters only i.e. once QR and SVD steps are done. Let $\bar{F}_{(.)}$ denotes the flop-count to estimate $\{A, B, C, D, K\}$ parameters. Therefore,
 
 \small
	\begin{equation*}\label{bdmatlab}\begin{aligned}
			\% Reduction =\frac{\bar{F}_{conv}- \bar{F}_{p}}{\bar{F}_{conv}} \times 100 = \frac{\mathcal{O}(n^2N)- \mathcal{O}(k^3)}{\mathcal{O}(n^2N)} \times 100
	\end{aligned}\end{equation*} 
\normalsize
	Since, $\bar{F}_{p}$ is independent of $N$, hence $\bar{F}_{p} < < \bar{F}_{conv}$. 
	\begin{table*}[h]
		\setlength\extrarowheight{2.5pt}
		\begin{center}
			\caption{Performance Comparison}
			\label{tab:Performance_comparison}
		\scalebox{0.95}{\begin{tabular}{|c|c|c|c|}
				\hline
				\textbf{Performance criterion} & \textbf{Conventional} & \textbf{Proposed (FR2SID)} & \textbf{Comments}          \\ \hline
				\textbf{Memory cost ($M$)}           & $ \mathcal{O}(kN)$           & $ \mathcal{O}\big(k{N}_{d}\big)$   & saving by a factor of $d$ \\ \hline
				\textbf{Flop-count ($F_{QR}$)}           & $  8k^2(m+p)^2N + \frac{16}{3}k^3(m+p)^3$           & $8k^2(q+1)(m+p)^2N+\frac{(48q+32)}{3}k^3(m+p)^3$   & see Note \ref{note:fastmatrixmult} \\ \hline
				\textbf{Flop-count ($F_{SVD}$)}           & $k^2N(4p^2+2mp) + 2k^3p^3$        & $ k^3(4pm^2+12p^2m+(34/3)p^3)$   & since, $k^3 << k^2N \hspace{0.15cm}\text{(significant speed-up)}$  \\ \hline
				\textbf{Flop-count ($F_{ABCDK}$)}           & $   \mathcal{O}(n^2N) $           & $   \mathcal{O}(k^3)$   & $\approx 100 \% \hspace{0.15cm}\text{(for large N)}$  \\ \hline
				\textbf{Data moved ($DM$)}           & $  \mathcal{O}(k^2N)$           & $ \mathcal{O}\big(kN\big)$   & \text{by an order of magnitude $k$} \\ \hline
			\end{tabular}}
		\end{center}
	\end{table*}
	
	\begin{note}\label{note:fastmatrixmult}
		Note that matrix multiplication can be parallelized easily (e.g. pp. 407 of \cite{tropp_2020}). While the flop-count for streaming data compression is of $\mathcal{O}\big((k^2+kl)N\big)$, since the high order of complexity occurs due to the matrix multiplication involved, it can be implemented efficiently using parallel algorithms. The computation can be further accelerated by using fast Hankel matrix-vector multiplication (see Table 1.2 pp. 4 \cite{Victor_FastHankel}) or by deploying techniques presented in \cite{Francois_FastMatMul}.
	\end{note}
	\begin{lemma}\label{Lemma_datamoved}
		For the proposed algorithm, $DM_{p}$ is of $\mathcal{O}\big((k+l)N\big)$.   Also, $DM_{p} < DM_{conv}$.
	\end{lemma}
	\begin{proof}
		Adding all DM terms from \eqref{DMpSDC}, \eqref{DMpRQR} and \eqref{DMpSVD} we get,
		\begin{equation}\label{DMproposed}\begin{aligned}
				DM_{p} &= DM_{SDC} + DM_{RQR} + DM_{SVD} + DM_{ABCDK}\\
				&= 2k(q+2)(m+p)N + lN +\text{terms having $k^2$ and $kl$}\\
				& \approx \big(2k(q+2)(m+p)+ l\big)N
		\end{aligned}\end{equation} 
		Since $k<<N$ therefore $k^2 << kN$. Hence we have only taken the dominating terms for final calculation in the last step. For conventional case, data movement turns out to be approximately $DM_{conv} = \mathcal{O}(k^2N)$ (see section \ref{section_main_issues}). Note that, DM for the proposed method is less than the DM for the conventional methods by an order of $k$. 
		Therefore, $DM_{p} < DM_{conv}$. 
	\end{proof}
	\begin{note}
		 Since, the latency term ($\alpha \times \#messages$) in \eqref{Talgo} may dominate other terms due to the fact that $\alpha >> \beta > \gamma$ \cite{henessy_latency} therefore data movement maybe more important for large data size.
	\end{note}
	Finally from Lemma \ref{Lemma_memorycost}, Lemma \ref{Lemma_flopcount} and Lemma \ref{Lemma_datamoved} it follows that $C_{p} < C_{conv}$. 
	
	\begin{theorem}\label{Theorem_detsto}
		For combined deterministic-stochastic subspace identification,  the algorithm cost for proposed method is less than the conventional algorithm i.e. $C_{p} < C_{conv}$.
	\end{theorem}
	
	A comparison for all the performance criteria for combined deterministic-stochastic identification is shown in Table \ref{tab:Performance_comparison}.  Since the oversampling parameter $l$ is very small (fixed at $l=5$), so we have ignored terms containing it. 
	
	\section{Efficiency and Accuracy of Randomized algorithm: QR and SVD steps}
	In this section, we experimentally demonstrate the efficiency and accuracy of the QR and SVD steps of Algorithm \ref{algo:fastsubspaceid}. All experiments were performed on intel core-i7 (9th generation) having level-2 cache of 2 MiB (fast memory), 32 GB RAM (slow memory) and 1 TB hard-drive using MATLAB-R2022a.
	We have defined the algorithm cost $C$ in \eqref{Algo_cost} which includes $\#messages$ and $\#words$. However these quantities cannot be measured directly in the experiments we perform next. On the other hand, since, we can measure $T_{algo}$ (see \eqref{Talgo}) easily using ``tic-toc" in MATLAB, we use $T_{algo}$ as a proxy for flops, $\#messages$ and $\#words$ in the numerical experiments below. Further, since $W$ is unknown, the optimal value of ${d}$ as in \cite{Demmel_seqQR} is also not known. Instead we choose ${d}$ empirically based on I/O data size.  
	
	First, we compare the current state-of-art SQR vs the proposed SDC with RQR. Then we show the effectiveness of the proposed range approximation algorithm (SVD step). For all the case studies, actual computation-time was measured experimentally using the ``tic-toc" command in MATLAB. Average computation time (ACT) was calculated by taking average time taken over 10 simulation runs where the random compression matrix is re-generated each time.
	
	\subsection{Full SQR vs SDC with RQR}
	
	In order to test the efficiency of the proposed SDC with RQR algorithm, matrices of different sizes ($H\in\mathbb{R}^{2k(m+p) \times N}$ with $2k(m+p) << N$, where each element is uniformly distributed between $(0,1)$) were generated. In all cases, the compression matrix $\mathcal{C}\in\mathbb{R}^{N\times N_{c}}$,
	where $N_{c}=2k(m+p)+l$ with $l=5$. This resulted in $\bar{H}:= (HH^T)^qH\mathcal{C}\in\mathbb{R}^{2k(m+p)\times N_{c}}$ with $q \in \{0, 1\}$. For fair comparison, $d$ (chosen heuristically) was kept same for SQR as well as for SDC with RQR. We have fixed $N=100,000$ in all cases while $\{k,m,p,d\}$ are varied for each case. From table \ref{tab:SQR_SDC_MSQR_Exp} we see that SDC with RQR is faster as compared to SQR and the advantage grows with larger data sizes.
	\begin{table}[h]
		\setlength\extrarowheight{1.5pt}
		\begin{center}
			\small\addtolength{\tabcolsep}{-5pt}
			\caption{ACT Comparison for Full SQR vs SDC with RQR}
			\label{tab:SQR_SDC_MSQR_Exp}
			\scalebox{0.8}{\begin{tabular}{|c|c|c|c|c|c|c|}
				\hline
				\textbf{\begin{tabular}[c]{@{}c@{}}Case\end{tabular}} & \textbf{\begin{tabular}[c]{@{}c@{}}Matrix \\ size \end{tabular}} & \textbf{\begin{tabular}[c]{@{}c@{}}Compressed \\matrix  size \end{tabular}} & \multicolumn{2}{c|}{\textbf{\begin{tabular}[c]{@{}c@{}}ACT speedup\\ R-factor  \end{tabular}}} \\  \hhline{|~|~|~|--|}
				$\{k, m, p, d\}$	& ($H$)&($\bar{H}$) & $q=0$ & $q=1$\\ \hline
				\textbf{\begin{tabular}[c]{@{}c@{}} $\{10,2,2,5\}$\end{tabular}}   & $80 \times 100k$  &   $80 \times 85$   & 1.71 &  1.22   \\ \hline
				\textbf{\begin{tabular}[c]{@{}c@{}}$\{20,5,5,10\}$\end{tabular}} &  $400 \times 100k$ &  $400 \times 405$ &  3.11  & 2.09   \\ \hline
				\textbf{\begin{tabular}[c]{@{}c@{}}$\{60,10,5,15\}$\end{tabular}}      &  $1800 \times 100k$&   $1800 \times 1805$ & 3.84  &  2.57 \\ \hline
				\textbf{\begin{tabular}[c]{@{}c@{}}$\{100,10, 10,20\}$\end{tabular}}      &  $4000 \times 100k$  & $4000 \times 4005$& 5.69&  3.86  \\ \hline
			\end{tabular}}
		\end{center}
	\end{table}
	For large matrices the advantage results mainly from the fact that the proposed method performs the QR factorization on a much smaller compressed matrix ($\bar{H}$). Moreover, SQR has to handle the intermediate $Q$-factors (write-cycle), while the proposed algorithm completely ignores the Q-factors, leading to lesser data data movement between slow and fast memory. Further, see Note \ref{note:fastmatrixmult} for possible advantages in parallel implementations of the matrix multiplication required in SDC.
	\subsection{Range approximation}
	Next, we test the accuracy of our proposed range-space approximation algorithm based on two criteria: (i) accuracy of rank-preservation 
	(see Theorem \ref{Theorem_image_presevation}) and (ii) the distance between two subspaces as in \cite{golub1995canonical} (see Chapter 6.4.3 of \cite{golub1995canonical} for more details).
	
	\textbf{Distance between two subspaces \cite{golub1995canonical}:} Let, $A\in\mathbb{R}^{m\times n}$
	and $B\in\mathbb{R}^{m\times l}$ and assume that $rank(A)\geq rank(B)=r$. Let, $A = Q_{A}R_{A}$ and $B = Q_{B}R_{B}$ such that $Q_{A}$ and $Q_{B}$ forms the basis for $\mathcal{R}(A)$ and $\mathcal{R}(B)$ respectively. Next let the SVD of $Q_{A}^{T}Q_{B} := Udiag(\sigma_{1}(A,B),\sigma_{2}(A,B),\hdots,\sigma_{r}(A,B))V^{T}$ where $\sigma_{i} := cos(\theta_{i})$, $0 \leq \theta_{i} \leq \pi/2$ $\forall i \in \{1, 2 , \hdots, r\}$. 
	Then the closeness of $\mathcal{R}(A)$ and $\mathcal{R}(B)$ can
	be measured by the following equation: $d(A,B)= sin(\theta_{max})$.

	In order to test the proposed rank preservation theorem, matrices of different sizes $\zeta = \bar{L}_{p}W_{p} \in\mathbb{R}^{kp\times N}$ with $kp << N$ and $\bar{\zeta} =\bar{L}_{p}\bar{W}_{p}\in\mathbb{R}^{kp\times N_{c}}$ were generated randomly as in previous subsection.
	\begin{table}[h]
		\setlength\extrarowheight{2pt}
		\begin{center}
			\small\addtolength{\tabcolsep}{-2pt}
			\caption{Subspace approximation}
			\label{tab:Subspace_approximation}
			\scalebox{0.8}{\begin{tabular}{|c|c|c|c|c|}
					\hline
					\textbf{\begin{tabular}[c]{@{}c@{}}Case\\ $\{k, m, p\}$\end{tabular}} & \textbf{\begin{tabular}[c]{@{}c@{}}$\zeta$\\$(kp \times N)$\end{tabular}} & \textbf{\begin{tabular}[c]{@{}c@{}}\textbf{$\bar{\zeta}$}\\$(kp \times N_{c})$\end{tabular}} & \textbf{\begin{tabular}[c]{@{}c@{}}Rank\\ preserved\end{tabular}} & \begin{tabular}[c]{@{}c@{}}\textbf{Distance:}\\ $d(\zeta, \bar{\zeta})$\end{tabular} \\ \hline
					\textbf{\begin{tabular}[c]{@{}c@{}} $\{10,2,2\}$\end{tabular}}   & $20 \times 100k$  &   $20 \times 85$   & Yes &  9.42e-08   \\ \hline
					\textbf{\begin{tabular}[c]{@{}c@{}}$\{20,5,5\}$\end{tabular}} &  $100 \times 100k$ &  $100 \times 405$ &  Yes  &   2.12e-07 \\ \hline
					\textbf{\begin{tabular}[c]{@{}c@{}}$\{60,10,5\}$\end{tabular}}      &  $300 \times 100k$&   $300 \times 1805$ & Yes  &  3.87e-07  \\ \hline
					\textbf{\begin{tabular}[c]{@{}c@{}}$\{100,10,10\}$\end{tabular}}      &  $1000 \times 100k$  & $1000 \times 4005$& Yes &  7.45e-07  \\ \hline
			\end{tabular}}
		\end{center}
	\end{table}
	We have fixed $N=100,000$ in all cases while $\{k,m,p\}$ is varied. From Table \ref{tab:Subspace_approximation}),
	it is observed that the rank is preserved in all cases. Also, the distance between the actual and
	approximated range-spaces turns-out to be negligible. Hence, Table \ref{tab:Subspace_approximation} verifies
	the range-space approximation/rank preservation theorem experimentally. In the above experiment, we have chosen $q=0$. Similar results were also achieved for $q=1$ (more robust).
	\begin{note}
		It is evident from the sizes of matrices $\zeta \in \mathbb{R}^{kp \times N}$ and $\bar{\zeta} \in \mathbb{R}^{kp \times N_{c}}$, that significant speed-up is obtained for the proposed method. Hence we omit reporting any efficiency data for the computation involved in Table \ref{tab:Subspace_approximation} here.
	\end{note}
	\section{CASE STUDIES: System Identification}\label{section_case_studies}
	In this section we verify the effectiveness of the proposed algorithm by identifying several synthetic and one practical systems. These case-studies will be evaluated on the basis of the following metrics:
	\begin{itemize}
		\item Normalized eigenvalue error (NEE):
		\begin{equation*}
			NEE := \sum_{i=1}^{n} \frac{|\lambda_{i} - \hat{\lambda}_{i}|^2}{|\lambda_{i}|^2}
		\end{equation*} 
		where, $\lambda_{i}$'s are the actual eigenvalues of the system and $\hat{\lambda}_{i}$'s are the estimated eigenvalues. For the proposed (randomized) case we have taken $\hat{\lambda}_{i}$ to be average of estimated eigenvalues over $n_{iter}$ simulation runs i.e. $\hat{\lambda}_{i} := \frac{1}{n_{iter}} \sum_{j=1}^{n_{iter}}\hat{\lambda}_{i}^{(j)} $.
		\item Mean Squared Error (MSE) on the validation data set:
		\begin{equation*}
			\text{MSE} = \sum_{i=1}^{p}\bigg(\frac{1}{N_v}\sum_{s=1}^{N_v} (y_{i}(s) - \hat{y}_{i}(s))^2\bigg) 
		\end{equation*} where, $s =\{1,2,....,N_v\}$, $y \in \mathbb{R}^{p}$ is the actual output, $\hat{y} \in \mathbb{R}^{p}$ is the predicted output and the subscript $i$ denotes $i^{th}$ component of the output. For the proposed algorithm, we use average MSE over $n_{iter}$ simulation runs i.e. Net-MSE (proposed) = $\frac{1}{n_{iter}}\sum_{j=1}^{n_{iter}} \text{MSE}_j$, where  ${MSE}_j$ is the MSE for the $j$-th run.
	\end{itemize}
	In all case-studies we have used input $u(t)$ and noise $e(t)$ as white gaussian signals. The variance of noise is decided based on the chosen SNR.
	\subsection{Synthetic models}
	In order to have wide separation between poles we have used the ``randi" function in Matlab to generate a fast pole $\lambda_{fast} \in ([-100,-50])$ and a slow pole $\lambda_{slow} \in (0.001*[-10,-1])$, while the remaining $(n-2)$ poles are placed randomly from a uniform distribution between $\lambda_{fast}$ and $\lambda_{slow}$. Thereafter, standard pole placement technique is used to place the eigenvalues of randomly generated system matrices $\{A,B,C,D,K\}$ (with iid normal elements) at the pole locations generated above. In all the case-studies we fix: $n_{iter} = 50$, oversampling-parameter $l=5$ and $N_{v} \approx 0.3*N_{t}$.
	\subsubsection{\textbf{Deterministic Case}} 
	We examine three distinct systems generated randomly. In these case studies, we have selected $q = 0$ due to the absence of noise.  The proposed method is approximately five times faster (see Table \ref{tab:Deterministic case}) as compared to the conventional methods with almost matching NEE and MSE.
	
	\subsubsection{\textbf{General Case}} 
	In this case we have taken four randomly generated systems. We explore the impact of SNR as well as resilience to noise ($q \in \{0, 1\}$) on the estimated parameters.  The performance comparison is presented in Table \ref{tab:General case}. Based on the data presented in the table, we can deduce the following:
	\begin{itemize}
		\item The ACT gap increases with increasing size of data matrices (high dimensional systems and large sample size). For example, in case 1, it is roughly 2.5 times faster, and in case 3, it's about 5.5 times faster compared to the conventional case.
		\item As the SNR decreases, the quality of estimates degrades for all methods.
		\item 	For low SNR case, the performance of the proposed algorithm with $q=0$ degrades faster as compared to other conventional algorithms (for instance, refer case 2(c)). This can be explained due to possibility of noise amplification during data compression. 
		\item The proposed algorithm, using a power-method approach ($q=1$), demonstrates strong noise robustness at a slightly increased computational cost.
		\item In case 4, Inf and NA indicates that the computer system have gone ``out-of-memory" implying that the conventional methods did not work in these cases.
	\end{itemize}
	In summary, the proposed method is very efficient for  large data/matrix sizes and has demonstrated good noise robustness, acceptable mean squared error (MSE), and faster estimation when compared to conventional algorithms. 
	\begin{table*}[htbp]
		\setlength\extrarowheight{2.25pt}
		\begin{center}
			\caption{Performance comparison for Deterministic Case}
			\label{tab:Deterministic case}
			\scalebox{0.967}{\begin{tabular}{|c|c|c|c|c|c|c|c|c|c|c|c|c|c|c|}
				\hline
				\textbf{S.No.} & \textbf{\{n, k, m, p, d\}} & \textbf{N}& \textbf{Nc} & \multicolumn{3}{c}{\textbf{ACT (msec.)}} & \multicolumn{3}{|c}{\textbf{NEE}}& \multicolumn{3}{|c|}{\textbf{Net-MSE}} \\
				\cline{5-13}
				&  & $(\times 10^3)$& & \textbf{FR2SID}& \textbf{N4SID} & \textbf{MOESP}& \textbf{FR2SID}& \textbf{N4SID} & \textbf{MOESP}& \textbf{FR2SID} & \textbf{N4SID} & \textbf{MOESP}\\  \hline
				\textbf{1}&\{2, 3, 5, 5, 5\}& 70 & 65 & 84 & 525& 410& 6.04e-25 & 1.85e-23&6.38e-24  & 3.05e-25 &1.66e-24 & 1.96e-25\\  \hline
				\textbf{2}& \{5, 6, 5, 5, 5\}& 70 & 125 &  128 & 589& 494&  7.33e-19& 6.24e-24 & 2.35e-23&1.39e-25  &3.45e-23 & 6.10e-26\\  \hline
				\textbf{3}&\{10, 11, 10, 10, 15  \} & 70  & 445 &  670& 3788& 3590& 1.89e-16& 9.12e-21 & 1.39e-23&  1.18e-23& 1.37e-23 &3.62e-24 \\  \hline  
			\end{tabular}}
		\end{center}
	\end{table*}
	
	\begin{table*}[htbp]
		\setlength\extrarowheight{2.25pt}
		\begin{center}
			\caption{Performance comparison for General Case}
			\label{tab:General case}
			\scalebox{0.887}{
				\begin{tabular}{|c|c|c|c|c|c|c|c|c|c|c|c|c|c|c|c|}
					\hline
					\textbf{S.No.} & \textbf{\{n, k, m, p, d\}} & \textbf{N} & \textbf{Nc}& \textbf{SNR} & \textbf{q}& \multicolumn{3}{c}{\textbf{ACT (msec.)}} & \multicolumn{3}{|c}{\textbf{NEE}}& \multicolumn{3}{|c|}{\textbf{Net-MSE}} \\
					\cline{7-15}
					& & $(10^3)$ & & & & \textbf{FR2SID}&\textbf{N4SID} & \textbf{MOESP}& \textbf{FR2SID} & \textbf{N4SID} & \textbf{MOESP}& \textbf{FR2SID}& \textbf{N4SID} & \textbf{MOESP}\\  \hline
					
					\textbf{1 (a)}& \{2, 5, 5, 5, 4\}&44 & 105&  100& 0& 45& 193& 179 &2.52e-08 &2.43e-12 & 1.91e-12&  2.14e-07& 2.81e-10 & 1.43e-10 \\  \hhline{|~|~|~|~|~|--|~|~|-|~|~|-|~|~|}
					& & & &  & 1& 67& & &7.57e-11 & & & 1.09e-08 &  & \\  \hline
					
					\textbf{1 (b)}& &- &- & 70 &0& -& -& -& 3.9e-04 & 1.16e-09 &1.10e-09& 1.89e-05&3.99e-06 & 6.47e-06 \\  \hhline{|~|~|~|~|~|--|~|~|-|~|~|-|~|~|}
					& & & &  &1&- & & &1.12e-04 & & & 3.16e-05&  & \\  \hline  
					
					\textbf{1 (c)}& &- &- & 50 &0& -& -& -&0.026 & 9.39e-08& 6.97e-08& 2.49e-02 &2.61e-04  & 7.20e-04\\  \hhline{|~|~|~|~|~|--|~|~|-|~|~|-|~|~|}
					& & & &  &1& -& & &1.25e-04& & &8.84e-03  &  & \\  \hline  
					
					\textbf{2 (a)}& \{10, 15, 5, 5, 10\} & 90 & 305&  100& 0&524 & 2682&2595 &2.68e-06 & 9.29e-10& 8.80e-10& 2.17e-08 & 7.29e-10 & 1.32e-09\\   \hhline{|~|~|~|~|~|--|~|~|-|~|~|-|~|~|}
					& & & &  & 1&832 & & & 2.51e-09& & & 9.52e-10 &  & \\  \hline
					
					\textbf{2 (b)}& &- & -&  70 &0& -& -& -& 4.84e-03& 1.61e-06&7.15e-07 & 3.80e-05 & 1.06e-06 & 5.22e-06\\     \hhline{|~|~|~|~|~|--|~|~|-|~|~|-|~|~|}
					& & & &  &1&- & & & 3.20e-04& & &5.24e-06  &  & \\  \hline  
					
					\textbf{2 (c)}& & - & -&  50  &0&- & -&- & 177.956& 6.1e-03&3.12e-05 &  1.52e-03& 2.33e-04 &7.87e-05 \\   \hhline{|~|~|~|~|~|--|~|~|-|~|~|-|~|~|}
					& & & &  &1&- & & & 0.7973& & &7.01e-04  &  & \\  \hline  
					
					\textbf{3 (a)}& \{30, 40, 10, 10, 25\}& 100 & 1605&  100  & 0& 3492& 21365 & 20893& 1.92e-06&7.30e-09 & 7.55e-09&1.22e-09  &3.26e-10  &1.99e-09 \\  \hhline{|~|~|~|~|~|--|~|~|-|~|~|-|~|~|}
					& & & &  & 1&5142 & & &1.35e-08 & & & 2.45e-10 &  & \\  \hline
					
					\textbf{3 (b)}& & - & -&  70 &0& -& -& -& 6.77e-01& 1.25e-05& 6.33e-06& 1.06e-06 & 7.65e-08 &8.67e-08 \\   \hhline{|~|~|~|~|~|--|~|~|-|~|~|-|~|~|}
					& & & &  &1&- & & & 2.08e-03& & & 6.34e-07 &  & \\  \hline  
					
						\textbf{4}& \{50, 70, 20, 20, 25\}& 150 & 5605&  100  & 0& 52688& \textcolor{red}{Inf} & \textcolor{red}{Inf} & 9.97e-07&\textcolor{red}{NA} & \textcolor{red}{NA} &5.63e-10  &\textcolor{red}{NA}  &\textcolor{red}{NA} \\  \hhline{|~|~|~|~|~|--|~|~|-|~|~|-|~|~|}
					& & & &  & 1&92561 & & &2.51e-08 & & & 3.40e-10 &  & \\  \hline
			\end{tabular}}
		\end{center}
	\end{table*}
	
	\subsection{Identification of Pressurized Heavy Water Reactor}
	In this section, we apply the proposed method to identify an LTI zone power model for a pressurized heavy water nuclear reactor (PHWR). This particular application was chosen since it  typically exhibits both fast and slow dynamics. In PHWR nuclear reactors the fastest time-constants are in the order of 0.05 seconds, while the slowest oscillations due to Xenon occur over 20 hours \cite{nuclear}. For more details about the system dynamics considered in this case study, the reader is referred to \cite{Vaswani_NuclearReactor}.  It is observed experimentally that the PHWR models have high noise sensitivity, resulting in the estimation of spurious poles with low values of $k$. Moreover, the poles are very close to the origin hence a small amount of noise may result in the estimated model being unstable. Hence, for this particular model, a slightly higher value of $k$, than those used in the synthetic cases above, is selected.
 
 The data is generated from a $56^{th}$ order MIMO zone power model of the PHWR with all real poles. The fastest pole $\lambda_{fastest} = -113.72$ while the slow dynamics is dominated by $\lambda_{slowest} = -2.88e-05$ (very close to origin). The parameters for this model are: $\{k, m,  p, d, l, N, N_{c} \}$  = \{120, 14, 15, 25, 10, 200000, 6970\}, $n_{iter} = 5,  N_{v} = 50000$, $SNR =100$. Model order ($n$) estimation is done by plotting log of singular values. A sharp knee is observed around the $35^{th}$ singular value. Consequently we have chosen $n=35$. The performance of the proposed algorithm is given in  Table \ref{tab:zonepower_case4}. In this case the $H$ matrix does not fit in the available RAM (out-of-memory). Hence conventional methods cannot be applied.

	\begin{table}[h]
		\begin{center}
		   \setlength\extrarowheight{2pt}
			\caption{Performance comparison for Zone power model}
			\label{tab:zonepower_case4}
			\scalebox{0.85}{
				\begin{tabular}{|c|c|c|c|c|}
					\hline
					\textbf{\begin{tabular}[c]{@{}c@{}}Performance\end{tabular}}  & \textbf{N4SID} & \textbf{MOESP}&  \multicolumn{2}{c|}{\textbf{\begin{tabular}[c]{@{}c@{}}FR2SID  \end{tabular}}} \\  \hhline{|~|~|~|--|}
					& && $q=0$ & $q=1$\\ \hline
					\textbf{ACT (msec)}                                                         & \textcolor{red}{Inf} & \textcolor{red}{Inf} & 158787   &   251485         \\ \hline
					\textbf{Net-MSE}                                                             &\textcolor{red}{NA} & \textcolor{red}{NA} & 2.136e-04 &     3.268e-05       \\ \hline
					\textbf{NEE}                                                                     & \textcolor{red}{NA} & \textcolor{red}{NA} & 0.079     &   0.0025       \\ \hline
			\end{tabular}}
		\end{center}
	\end{table} 

	\section{CONCLUSION}
	
	A novel fast randomized subspace identification algorithm to identify combined deterministic-stochastic LTI state-space model has been presented. The proposed algorithm is able to outperform the conventional subspace methods in terms of memory cost, flop-count and computation-time, with comparable accuracy, for cases where conventional methods can still accommodate large data sizes. However, the proposed method is capable of handling significantly larger data sizes than what can be processed in conventional methods. It seems that the degradation in the estimates for increased noise, is amplified due to compression. The effect of the compression matrix on this degradation and the design of a compression technique with provable noise immunity, are currently under investigation.
	

	


	\addtolength{\textheight}{-2cm}   
	





\begin{thebibliography}{99}
		
		\bibitem{Ljung_book} Ljung, L. System Identification: Theory for the User (second edition). Prentice Hall, Upper Saddle River, New Jersey 1999.
		\bibitem{Storage} Goda, Kazuo, and Masaru Kitsuregawa. ``The history of storage systems." Proceedings of the IEEE 100.Special Centennial Issue (2012): 1433-1440.
		\bibitem{tropp_2020}Martinsson, Per-Gunnar, and Joel A. Tropp. ``Randomized numerical linear algebra: Foundations and algorithms." Acta Numerica 29 (2020): 403-572.
		\bibitem{lapack} `` LAPACK - Linear Algebra Package," https://www.netlib.org/lapack/.
		\bibitem{blas} `` BLAS - Basic Linear Algebra Subprograms,"  https://www.netlib.org/blas/.
		\bibitem{Schoukens_explength_2015} J. Schoukens and S. Kolumban, ``Study of the minimum experiment length to identify linear dynamic systems: A variance based approach," 2015 IEEE International Instrumentation and Measurement Technology Conference (I2MTC) Proceedings, 2015, pp. 963-968, doi: 10.1109/I2MTC.2015.7151400.
		\bibitem{nuclear} Chakraborty, Abhishek, Suneet Singh, and M. P. S. Fernando. ``A novel approach for bifurcation analysis of out of phase xenon oscillations using multipoint reactor kinetics." Nuclear Engineering and Design 328 (2018): 333-344.
		\bibitem{blastfurnace} Gao, Chuanhou, Jiusun Zeng, and Zhimin Zhou. ``Identification of multiscale nature and multiple dynamics of the blast furnace system from operating data." AIChE journal 57, no. 12 (2011): 3448-3458.
		\bibitem{distillation_column} Vora, Nishith, and Prodromos Daoutidis. ``Dynamics and control of an ethyl acetate reactive distillation column." Industrial \& engineering chemistry research 40, no. 3 (2001): 833-849.
		\bibitem{battries} Hu, Yiran, and Yue-Yun Wang. ``Two time-scaled battery model identification with application to battery state estimation." IEEE Transactions on Control Systems Technology 23, no. 3 (2014): 1180-1188.
		\bibitem{ljung2010perspectives}Ljung, Lennart. ``Perspectives on system identification." Annual Reviews in Control 34.1 (2010): 1-12.
		\bibitem{HoKalman} Ho, B.L., and Rudolf E. Kalman. ``Effective construction of linear state-variable models from input/output functions." at-Automatisierungstechnik 14.1-12 (1966): 545-548. 
		\bibitem{SubspaceId_book} Van Overschee, Peter, and BL De Moor. Subspace identification for linear systems: Theory- Implementation-Applications. Springer Science \& Business Media, 2012.
		\bibitem{CVA_Larimore} Larimore, Wallace E. ``Canonical variate analysis in identification, filtering, and adaptive control." 29th IEEE Conference on Decision and control. IEEE, 1990.
		\bibitem{MOESP_Verhaegan} Verhaegen, Michel. ``Identification of the deterministic part of MIMO state space models given in innovations form from input-output data." Automatica 30.1 (1994): 61-74.
		\bibitem{N4SID_Overschee_detsto} Van Overschee, Peter, and Bart De Moor. ``N4SID: Subspace algorithms for the identification of combined deterministic-stochastic systems." Automatica 30.1 (1994): 75-93.
		\bibitem{Cho_Kailath_1994} Cho, Young Man, Guanghan Xu, and Thomas Kailath. ``Fast identification of state-space models via exploitation of displacement structure." IEEE Transactions on Automatic Control 39.10 (1994): 2004-2017.
		\bibitem{Cho_Kailath_1995} Cho, Young Man, and Thomas Kailath. ``Fast subspace-based system identification: An instrumental variable approach." Automatica 31.6 (1995): 903-905.
		\bibitem{Sima_2004} Sima, Vasile, Diana Maria Sima, and Sabine Van Huffel. ``High-performance numerical algorithms and software for subspace-based linear multivariable system identification." Journal of computational and applied mathematics 170.2 (2004): 371-397.
		\bibitem{Dohler_2012} Dohler, Michael, and Laurent Mevel. ``Fast multi-order computation of system matrices in subspace-based system identification." Control Engineering Practice 20.9 (2012): 882-894.
		\bibitem{Sima_2001} Mastronardi, Nicola, et al. ``A fast algorithm for subspace state-space system identification via exploitation of the displacement structure." Journal of Computational and Applied Mathematics 132.1 (2001): 71-81.
		\bibitem{Katayama_LQfullmodel} Katayama, Tohru. ``Subspace identification of combined deterministic-stochastic systems by LQ decomposition." Proceedings of the 2010 American Control Conference. IEEE, 2010.
		\bibitem{Demmel_seqQR} Demmel, James, et al. ``Communication-optimal parallel and sequential QR and LU factorizations." SIAM Journal on Scientific Computing 34.1 (2012): A206-A239.
		\bibitem{Demmel_report} Demmel, James, et al.  ``Communication-optimal parallel and sequential QR and LU factorizations: theory and practice." (2013). https://doi.org/10.48550/arXiv.0806.2159
		\bibitem{mahoney_2011}Mahoney, Michael W. ``Randomized algorithms for matrices and data." Foundations and Trends® in Machine Learning 3.2 (2011): 123-224.
		\bibitem{halko_2011}Halko, Nathan, Per-Gunnar Martinsson, and Joel A. Tropp. ``Finding structure with randomness: Probabilistic algorithms for constructing approximate matrix decompositions." SIAM review 53.2 (2011): 217-288.
		\bibitem{kramer2018_CURHankel}Kramer, Boris, and Alex A. Gorodetsky. "System identification via CUR-factored Hankel approximation." SIAM Journal on Scientific Computing 40.2 (2018): A848-A866.
		\bibitem{minster2021_ERARSVD}Minster, Rachel, et al. ``Efficient algorithms for eigensystem realization using randomized SVD." SIAM Journal on Matrix Analysis and Applications 42.2 (2021): 1045-1072.
		\bibitem{anderson2022_ESRSVD} Wang, Han, and James Anderson. "Large-scale system identification using a randomized svd." 2022 American Control Conference (ACC). IEEE, 2022.
		\bibitem{Vatsal_Fastsub} Kedia, Vatsal, and Debraj Chakraborty. "Fast Subspace Identification for Large Input-Output Data." 2022 American Control Conference (ACC). IEEE, 2022.
		\bibitem{Vatsal_RandSID} Kedia, Vatsal, and Debraj Chakraborty. "Randomized Subspace Identification for LTI Systems." 2023 European Control Conference (ECC). IEEE, 2023.
		\bibitem{LjungSID} Qin, S. Joe, Weilu Lin, and Lennart Ljung. "A novel subspace identification approach with enforced causal models." Automatica 41.12 (2005): 2043-2053.
		\bibitem{Katayama_book}Katayama, Tohru. Subspace methods for system identification. Vol. 1. London: Springer, 2005.
		\bibitem{Ljung_MATLAB}Ljung, Lennart. ``Aspects and experiences of user choices in subspace identification methods." IFAC Proceedings Volumes 36.16 (2003): 1765-1770.
		\bibitem{MOESP_bd_Verhaegan}Verahegen, M., and Patrick Dewilde. ``Subspace model identification. part i: The output-error state-space model identification class of algorithm." Int. J. Control 56 (1992): 1187-1210.
		\bibitem{MOESP_Verhaegen_stochastic} Verhaegen, Michel. ``Identification of the deterministic and stochastic part of MIMO state space models under the presence of process and measurement noise." European control conference (1993): 1313-1318.
		\bibitem{bryc_rotation} W. Bryc,``Rotation invariant distributions,” in The Normal Distribution, pp. 51–69, Springer, 1995.
		\bibitem{normal_rotation} S. O. Gharan, ``CSE 521: Design and analysis of algorithms I, lecture
		7,” 2018. Last accessed 16 September 2022.
		\bibitem{feng_randomrank}Feng, Xinlong, and Zhinan Zhang. ``The rank of a random matrix." Applied mathematics and computation 185.1 (2007): 689-694.
		\bibitem{henessy_latency} Hennessy, John L., and David A. Patterson. Computer architecture: a quantitative approach. Sixth edition. Elsevier, 2019.
		\bibitem{Zhang_perturbsvd_2022} Zhang, Yichi, and Minh Tang. "Perturbation Analysis of Randomized SVD and its Applications to High-dimensional Statistics." arXiv preprint arXiv:2203.10262 (2022).
		\bibitem{Victor_FastHankel} Pan, Victor. Structured matrices and polynomials: unified superfast algorithms. Springer Science \& Business Media, 2001.
		\bibitem{Francois_FastMatMul} Gall, Francois Le, and Florent Urrutia. "Improved rectangular matrix multiplication using powers of the coppersmith-winograd tensor." Proceedings of the Twenty-Ninth Annual ACM-SIAM Symposium on Discrete Algorithms. Society for Industrial and Applied Mathematics, 2018.
		\bibitem{golub1995canonical}Golub, Gene H., and Charles F. Van Loan. Matrix computations. JHU press, 2013.
		\bibitem{Vaswani_NuclearReactor} Vaswani, P. D., et al. ``Optimised structured state feedback controller for zone power and bulk power control of PHWRs." Annals of Nuclear Energy (2021): 108835.
		\bibitem{Memoryhierarchy} Memory Hierarchy - Chemeketa CS160 Reader. Retrieved September 8, 2022, from http://computerscience.chemeketa.edu/cs160Reader/
		ComputerArchitecture/MemoryHeirarchy.html. 
		
		
	\end{thebibliography}
\end{document}